\documentclass[a4paper]{article}

\usepackage[utf8]{inputenc}  
\usepackage[T1]{fontenc}
\usepackage{palatino}
\usepackage{graphicx}
\usepackage[margin=1in, top=1in, bottom=1in]{geometry}
\usepackage{comment}
\usepackage{pifont}

\usepackage{xcolor}
\definecolor{modification}{rgb}{0, 0, 0}

\usepackage{amsmath, amsthm, mathtools} 
\usepackage{thmtools}
\usepackage{thm-restate}
\usepackage{amssymb, amsfonts}
\usepackage[final]{hyperref} 
\usepackage{xspace} 
\usepackage{xcolor} 

\hypersetup{colorlinks=true,
            allcolors=blue,
            citecolor=blue!75!black
            }

\usepackage[obeyDraft, textsize=small, colorinlistoftodos]{todonotes}

\usepackage{qcircuit}
\usepackage{adjustbox}
\usepackage{multirow}
\usepackage[figuresright]{rotating}

\usepackage{mathrsfs} 
\usepackage{bbm} 
\usepackage{url}

\usepackage[eprint,          
            doi=false,
            url=false,
            backend = biber,
            sorting = nyt,
            maxnames=10,     
            giveninits=true, 
            safeinputenc,     
           style = nature,
           isbn = false,
           date=year]{biblatex}
\theoremstyle{plain} 
\newtheorem{thm}{Theorem} 
\newtheorem{lem}{Lemma}

\newcounter{scount}[subsection] 

\theoremstyle{definition} 
\newtheorem{defn}{Definition} 
\theoremstyle{remark} 
\newtheorem*{rem*}{Remark} 

\global\long\def\ket#1{|#1\rangle}%
\global\long\def\kb#1#2{|#1\rangle\langle#2|}%
\global\long\def\Hs{\mathcal{H}}%
\global\long\def\id{\mathbb{I}}%
\global\long\def\tr{\text{Tr}}%

\global\long\def\Luders{\text{L{\"u}ders }}

\global\long\def\urule{\omega}

\global\long\def\uruleA{\urule}

\global\long\def\uruleL{\urule^{\mathsf{L}}}
\global\long\def\uruleP{\urule^{\mathsf{P}}}

\definecolor{mycolour}{rgb}{0, 0, 0}
\definecolor{mycolour2}{rgb}{0, 0, 0}
\definecolor{red}{rgb}{255, 0, 0}

\DeclareRobustCommand{\loongrightarrow}{%
  \DOTSB\relbar\joinrel\relbar\joinrel\relbar\joinrel\rightarrow
}
\DeclareRobustCommand{\loongmapsto}{\DOTSB\mapstochar\loongrightarrow}

\newcommand{\ignore}[1]{}

\title{Beyond the Projection Postulate and Back:\\ Quantum Theories with Generalised State-Update Rules}
\author{Vincenzo Fiorentino and Stefan Weigert\\
\small vincenzo.fiorentino@york.ac.uk, stefan.weigert@york.ac.uk \\}

\date{%
    \normalsize  Department of Mathematics, University of York, York YO10 5GH, United Kingdom\\[2ex]%
    January 2026\footnote{Originally uploaded 6 June 2025; this is the version published as \href{https://doi.org/10.1103/2zpm-jsh7}{Phys. Rev. A 113 012204 (2026)}.} 
}

\begin{document}

\maketitle

\begin{abstract}
Are there consistent and physically reasonable alternatives to the projection postulate? Does it have unique properties compared with acceptable alternatives? We answer these questions by systematically investigating hypothetical state-update rules for quantum systems that Nature could have chosen over the \Luders rule. Among other basic properties, any prospective rule must define unique post-measurement states and not allow for superluminal signalling. Particular attention will be paid to consistently defining post-measurement states when performing local measurements in composite systems. Explicit examples of valid unconventional update rules are presented, each resulting in a distinct, well-defined foil of quantum theory. This framework of state-update rules allows us to identify operational properties that distinguish the projective update rule from all others and to put earlier derivations of the projection postulate into perspective.
\end{abstract}

\maketitle

\tableofcontents

\section{Introduction}

The main goal of \textcolor{modification}{our} contribution is to identify operational properties that distinguish the quantum-mechanical projection postulate from \textcolor{modification}{hypothetical but} physically reasonable state-update rules.  
\textcolor{modification}{Not surprisingly, attempts to justify \Luders rule (for which both Dirac and von Neumann deserve credit)\footnote{\textcolor{modification}{In view of the historic developments, it would be appropriate to speak of the ``Dirac–von Neumann-Lüders rule'' \cite{sudbery_whose_2024} but we will continue using ``\Luders rule'' for simplicity. However, the projection postulate as formulated by Dirac \cite{dirac_principles_1930} and Lüders \cite{luders_uber_1950} differs from von Neumann's version \cite{vonneumann_mathematische_1932} in key aspects; see Sec.\ \ref{sec: ruled out theories}.}}} \textcolor{modification}{have a long history, and a recent debate revolves around the claim that it can be derived from the other postulates of quantum theory (cf.\ Sec.\ \ref{Qmstateupdaterule}).}
Many of these approaches have in common that their focus is on non-composite systems, either ignoring measurements on composite systems or making implicit assumptions about their consequences.

\textcolor{modification}{Our main tool will be the construction of \textit{foil theories} of quantum theory (cf.\ Sec.\ \ref{subsec: foil theories}), i.e.\ toy theories that help ``to highlight the distinctive characteristics of quantum theory by contrasting with it'' (adapted from \cite{chiribella_quantum_2016})}. \textcolor{modification}{Thus, we investigate the nature of the projection postulate by comparing quantum theory with similar theories that are equipped with different }\textit{generalised state-update rules} (GURs). In other words, we introduce \textit{generalised state-update theories} (GUTs) for which the post-measurement states associated with a given outcome are not necessarily those prescribed by the \Luders rule. From the outset, we will take into account both single and composite systems. 
Alternative post-measurement rules have been considered occasionally (cf.\ Sec.\ \ref{Qmstateupdaterule}), but we are not aware of any systematic approach.

To qualify as a legitimate update rule, the assignment of post-measurement states must satisfy a number of  well-motivated physical requirements. The framework we present is broad enough to accommodate the standard quantum-mechanical projection postulate as well as many hypothetical rules, including the so-called \textit{passive measurements} that leave the state of a system untouched while still producing outcome probabilities governed by Born's rule. 

In this setting, one can easily compare the properties of different update rules and, importantly, seek to identify those that are unique to the \Luders rule. It also provides a suitable backdrop to assess existing derivations of  the \Luders rule summarised in Appendix \ref{sec:historical review Luders}. 

This paper is organised as follows. In Sec.\ \ref{background} we set the scene by briefly reviewing earlier discussions related to the quantum mechanical state-update rule and describe the idea of quantum foil theories. Sec.\ \ref{sec: AMT Framework} presents a simple axiomatic formulation of quantum theory and a detailed analysis of the state-update in composite quantum systems. In Sec.\ \ref{AURules}, we describe the basic requirements that state-update rules must satisfy and how they give rise to quantum foil theories. Explicit examples of legitimate and invalid update rules are given in Sec.\ \ref{sec: examples}. In Sec.\ \ref{sec: singling out Lueders} we explain what distinguishes the quantum-mechanical projection postulate from other valid update rules. In the concluding section, we summarise our results, discuss them, and outline future research.

\section{Background and motivation} \label{background}
\subsection{The quantum-mechanical state update rule}\label{Qmstateupdaterule}
Using quantum systems as information carriers has led to applications without classical equivalents. Quantum measurements, in particular, provide a resource that contributes, for example, to secure communication protocols \cite{bennett_quantum_2014}, quantum teleportation \cite{bennett_teleporting_1993} device-independent random number generators \cite{colbeck_quantum_2009}, and to speedy algorithms when using measurement-based quantum computation \cite{briegel_measurement-based_2009}.

The advantages of quantum measurements arise from fundamental differences when compared with measurements in classical theories. Classical measurements are deterministic in that they reveal pre-existing values of observables, and the effect on the state of the observed system is negligible, at least in principle. In quantum theory, measurement outcomes are probabilistic, and the state of a system after a measurement coincides with the pre-measurement state only in exceptional cases. Phenomenologically, the quantum mechanical state update is governed by the \textit{projection postulate}, which describes a ``collapse'' of the state to an eigenstate of the operator representing the observed outcome \textcolor{modification}{\cite{dirac_principles_1930, vonneumann_mathematische_1932, luders_uber_1950}}. The collapse, also known as \textit{state reduction} or \textit{\Luders rule} \cite{luders_uber_1950}, ensures \textit{deterministic repeatability}: two measurements of an observable always yield identical outcomes when implemented in quick succession \cite{compton_directed_1925}.

The standard postulates of quantum theory neither explain whether the state-update is a physical process \cite{bohm_suggested_1952, dewitt_many-worlds_2015}, nor do they define what constitutes a measurement \cite{bell_against_1990, bacciagaluppi_heisenberg_2009}. The foundational questions resulting from this dilemma are often grouped under the term ``\textit{measurement problem}'' \cite{busch_quantum_1996, norsen_foundations_2017}. The present work does \textit{not} attempt to resolve the interpretational challenges posed by this long-standing problem.

\textcolor{modification}{Derivations} of the \Luders rule have been given \cite{bell_moral_1966, herbut_derivation_1969, hellmann_quantum_2016, kostecki_luders_2014, martinez_search_1990} (see Appendix \ref{sec:historical review Luders} for a survey), and the difference between von Neumann's and Lüders' formulations of the projection postulate has been discussed repeatedly \cite{luders_uber_1950, busch_luders_1998, sudbery_whose_2024}. 
Some authors have argued that the projection postulate follows from the other standard postulates of quantum theory \cite{ozawa_quantum_1984, masanes_measurement_2019} and the debate of a recent derivation is ongoing \cite{stacey_masanes-galley-muller_2022, galley_reply_2022, kent_measurements_2023, masanes_response_2025, stacey_contradictions_2024}.

Occasionally, alternative post-measurement rules have been considered, either in the context of standard quantum theory \cite{kent_nonlinearity_2005, kent_quantum_2021, masanes_measurement_2019} or of generalised probabilistic theories \cite{kleinmann_sequences_2014}. A systematic investigation, as proposed here, has apparently not been carried out before.

\subsection{Foil theories} \label{subsec: foil theories}
Our approach is inspired by the constructive role of foil theories for the understanding of quantum theory \textcolor{modification}{\cite{chiribella_quantum_2016}}.  The strategy in defining foil theories is to embed quantum theory in a larger set of theories by relaxing specific assumptions about the structure of the state space, say, or the form of the time evolution.  We briefly describe typical examples of foil theories and point out that a systematic investigation of generalised state-update rules has, apparently, not yet been carried out.

\textit{Generalised probabilistic theories} (GPTs) \cite{barrett_information_2007, hardy_quantum_2001, janotta_generalized_2014}, for example, represent a large and versatile collection of foil theories. They are defined by their own sets of states, observables, and rules for system composition, all of which may differ from those of standard quantum theory. GPTs have become a standard tool to explore and identify distinctive characteristics of quantum theory which they contain as a special case. Many GPTs share properties such as the no-cloning theorem, uncertainty relations, or teleportation with quantum theory \cite{barnum_cloning_2006, oppenheim_uncertainty_2010}. Some of them reproduce the standard quantum-mechanical correlations \cite{wright_gleason-type_2019}, while others exhibit superquantum correlations, e.g. Boxworld \cite{popescu_quantum_1994}. 

\textit{Ontological models} describe a conceptually different class of foil theories not necessarily based on the familiar Hilbert space structure of quantum theory. They aim to explain its predictions in terms of an observer’s incomplete knowledge about the \textit{ontic state} that provides an objective and complete description of the state of a system. In these models, non-classical features of quantum theory such as complementarity of observables, superdense coding, interference phenomena, or uncertainty relations may result from \textit{epistemic restrictions} \cite{spekkens_evidence_2007, bartlett_reconstruction_2012}. 

Other foil theories differ from quantum theory in a single key feature only, such as attempts to modify \textit{Born’s rule} \cite{galvan_generalization_2008,galley_classification_2017}, or the \textit{tensor product} for system composition \cite{erba_composition_2024}. The consequences of associating observables with \textit{non-Hermitian operators} have been investigated, using \textit{PT-symmetric} \cite{bender_pt_2024} or \textit{normal} \cite{roberts_observables_2018} operators, for example. Quantum-like theories defined over the \textit{real numbers} \cite{stueckelberg_quantum_1960, renou_quantum_2021} or \textit{quaternions} \cite{finkelstein_foundations_1962}  rather than the field of complex numbers have also been studied.

\textit{Nonlinear generalisations of Schr{\"o}dinger's equation} have been suggested \cite{beretta_quantum_1984, abrams_nonlinear_1998, ferrero_nonlinear_2004, kent_nonlinearity_2005, rembielinski_nonlinear_2020}, most notably in Weinberg's work \cite{weinberg_testing_1989}. Gisin et al.\ \cite{gisin_weinbergs_1990, simon_no-signaling_2001,bassi_no-faster-than-light-signaling_2015} pointed out that they imply violations of the no-signalling principle, although some specific nonlinear transformations were later found not to be ruled out by Gisin’s argument \cite{czachor_nonlocal_1998, kent_nonlinearity_2005, ferrero_nonlinear_2004, helou_extensions_2017, rembielinski_nonlinear_2020}.

In the specific context of modifications to the projection postulate, ``\textit{causal quantum theory}'' \cite{kent_causal_2005} has been proposed, positing that the \Luders collapse is a well-defined physical process that satisfies strict local causality. Additionally, Kent introduced a new type of hypothetical measurement that is capable of revealing complete or partial information about the ``\textit{local}'' state of a system, without actually inducing a state update \cite{kent_nonlinearity_2005, kent_quantum_2021, kent_measurements_2023}, an idea first presented in \cite{busch_is_1997}. Such a measurement would, effectively, allow for yet other nonlinear state transformations without enabling superluminal signalling.

\section{Quantum theory and the projection postulate} \label{sec: AMT Framework}

\subsection{Axioms of quantum theory} \label{sec: axioms QT}

The foil theories we introduce will differ from quantum theory only in their \textit{update rule} assigning post-measurement states. The other standard postulates of non-relativistic, finite-dimensional quantum theory remain unchanged.
\begin{itemize}
    \item[($\mathsf{S}$)] The \textit{states} of a quantum system correspond to density operators $\rho \in \mathcal{S}(\Hs)$, non-negative operators with unit trace acting on a finite-dimensional, separable and complex Hilbert space $\Hs$ associated with the system.
        \item[($\mathsf{T}$)] The \textit{time evolution} of states, as well as any other reversible transformation, is described by a unitary operator $U\in \mathcal{U}(\Hs)$.
    \item[($\mathsf{C}$)] The state space of a \emph{composite} system is obtained from tensoring the state spaces describing its constituents, i.e. by $\mathcal{S}_{AB}=\mathcal{S}(\Hs_A \otimes \Hs_B)$. Given a state $\rho_{AB}$ of the composite system, the state of subsystem $A$ is obtained by tracing out the other one, $\rho_A = \text{Tr}_B \left[ \rho_{AB} \right]$, etc.   
    \item[($\mathsf{O}$)] \emph{Observable} quantities are represented by collections $M= \{P_x\}_{x=1}^n$ of mutually orthogonal projectors $P_x\in \mathcal{P}(\Hs)$, $P_xP_y=P_x \delta_{xy}$, that sum to the identity operator on $\Hs$, i.e.\ $\sum_x P_x=\id$. Each projector $P_x$ represents a possible outcome of a measurement of the observable $M$.
    \item[($\mathsf{B}$)] The \emph{probability} that a measurement of $M$ performed on a system in the state $\rho$ yields the outcome represented by $P_x$ is given by the \emph{Born rule}: $\text{prob}(x)=\text{Tr} \left[ \textcolor{modification}{P_x} \, \rho \right]$.
\end{itemize}

Finally, the \emph{projection postulate} describes what happens to the state of a quantum system when an observable is measured.

\begin{itemize}
    \item[($\mathsf{L}$)] If a measurement of $M$ yields an outcome represented by the projection $P_x$, then the initial state $\rho \in \mathcal{S}(\Hs)$ will update in the following way,
    \begin{equation} \label{eq: normalised collapsed state}
       \rho \; \stackrel{x}{\longmapsto} \; \rho_x = \frac{P_x \rho P_x}{\text{Tr} \left[ P_x \rho P_x \right] } \in \mathcal{S}(\Hs)\, .
    \end{equation}
\end{itemize}

The focus of our work is to investigate the role of postulate ($\mathsf{L}$) by replacing it with alternative update rules. 
In view of later developments, it will be convenient to suppress the normalisation factor in Eq. \eqref{eq: normalised collapsed state}. We introduce the map 
\begin{equation} \label{Lueders}
     \rho \; \stackrel{x}{\longmapsto} \uruleL (P_x, \rho) = P_x \rho P_x \,,
\end{equation}
which sends $\rho$ to a sub-normalized state\footnote{A \emph{sub-normalised} state is a non-negative operator on a Hilbert space with trace less than or equal to $1$, i.e.\ $\rho \in \bar{\mathcal{S}}(\Hs) = \{ \lambda \rho : \lambda \in [0,1], \rho \in \mathcal{S}(\Hs) \}$.} in the space $\bar{\mathcal{S}}(\Hs)$, and call it the \textit{\Luders  rule}, giving postulate ($\mathsf{L}$) its name. The trace of the post-measurement state, $\text{Tr} (\uruleL (P_x, \rho))=\text{Tr}(P_x \rho)$, yields the outcome probability as prescribed by the Born rule of postulate ($\mathsf{B}$). In addition, the maps  \eqref{eq: normalised collapsed state} and \eqref{Lueders} do not require the input state to be normalised. Replacing $\rho $ by $\lambda \rho$, $\lambda \in [0,1]$, shows that the \Luders rule is \textit{1-homogeneous}:  $\uruleL (P_x, \lambda \rho)=\lambda \uruleL (P_x, \rho) $.

The \textit{alternative update rules} $\urule (P_x, \rho)$ considered in Sec.\ \ref{sec: examples} will be modelled on the \textit{\Luders rule} in the sense that they are 1-homogeneous, only depend on the measurement outcome $P_x$ and the initial state $\rho$, and agree with the Born rule for outcome probabilities.\footnote{Postulate $(\mathsf{B})$ is understood to assign outcome probabilities to each \textit{individual} measurement. Consequently, Born's rule on its own does not account for \textit{correlations} between outcomes of \textit{multiple} measurements. Instead, the \textit{update rule} is also responsible for the joint probabilities for the outcomes of measurements carried out sequentially, or by distinct parties in composite systems. This subtle point proves significant when examining modifications of the \Luders rule $(\mathsf{L})$; see Appendix \ref{sec:historical review Luders} for further discussion.}
Further physically motivated properties that any \textit{acceptable} update rule should satisfy will be discussed in Sec.\ \ref{AURules}. 

Before defining update rules replacing ($\mathsf{L}$), we must take a closer look at the effect of quantum measurements performed on composite systems. As it stands, the rule ($\mathsf{L}$) does, in fact, \textit{not} specify the post-measurement state of a composite system when measurements are performed on \textit{subsystems}. The post-measurement state traditionally attributed to a quantum system relies on an additional implicit assumption. 

\subsection{Measurements on composite quantum systems} \label{LuedersOnCompSys}

The projection postulate ($\mathsf{L}$) does not explicitly cover \textit{local} measurements which arise in \textit{composite} quantum systems, i.e.\ measurements that are carried out on subsystems. To discuss this situation in detail, we will consider a system with \textit{two} constituents $A$ and $B$, say, and an associated Hilbert space $\Hs_{AB} = \Hs_A \otimes \Hs_B$, in line with axiom ($\mathsf{C}$). 

Suppose we wish to measure the \textit{global} observable $M_{AB}$, represented by an operator acting on the product space $\Hs_{AB}$.  To this  effect, a \textit{global} device $\mathcal{D}_{AB}$ must be used that interacts with the composite system in its entirety.   The outcome of the measurement is represented by one of the projection operators $P_x^{AB}$ (associated with $M_{AB}$), that sum to the identity on the space $\Hs_{AB} $, as described in axiom ($\mathsf{O}$). The probability to obtain any of these outcomes is found from axiom ($\mathsf{B}$), and the state $\rho_{AB}$ of the overall system will update according to the \Luders rule ($\mathsf{L}$) given in Eq.\ (\ref{eq: normalised collapsed state}). The projectors $P_x^{AB}$ do not have to respect the product structure of the space $\Hs_{AB}$: for a two-qubit system, they may project onto the states of the Bell basis, for example.  At no point do global measurements refer to the product structure of the space $\Hs_{AB}$.

However, on a bipartite quantum system, one may also carry out \textit{local} measurements. They arise naturally when its constituents are in separate locations. Local measurements are implemented by measuring devices $\mathcal{D}_{A}$ or $\mathcal{D}_{B}$ that act on subsystems only, i.e.\ the constituents $A$ or $B$, respectively. The postulate ($\mathsf{L}$) does not cover this situation: no post-measurement state $\rho_{AB,x}\in \mathcal{S}_{AB} \equiv \mathcal{S}(\Hs_A \otimes \Hs_B)$ is specified by ($\mathsf{L}$) when using the device $\mathcal{D}_{A}$ to carry out a measurement of the local observable $M_A$, resulting in the outcome $x$. Let us explain why.

According to axiom ($\mathsf{O}$),  the local observable $M_A$ corresponds to a collection of mutually orthogonal projectors $\{P_x^A\}_{x=1}^n$,  defined on $\Hs_A$ and summing to the identity,  $\sum_x P_x^A=\id_A$.  Thus, we associate measurement outcomes of $M_A$ with projectors $P_x^A$.  However, neither Eq.\ (\ref{eq: normalised collapsed state}) nor (\ref{Lueders}) can be used to determine the post-measurement state of the system as a whole: to apply  \Luders rule, we need to use (global) projectors $P_x^{AB}$ defined on $\Hs_{AB}$ but the (local) operators $P_x^A$ act on $\Hs_A$ only. 

For the experimenter carrying out a local measurement on subsystem $A$, the existence of the other constituent $B$ is irrelevant. 
If the composite system resides in the state $\rho_{AB}$, we invoke axiom ($\mathsf{C}$): the initial state of subsystem $A$ is given by the density matrix $\rho_A$ obtained by tracing out subsystem $B$. Then, we are in a position to apply ($\mathsf{L}$) to find the post-measurement state of subsystem $A$ upon obtaining the local outcome $P_x^A$, namely
\begin{equation} \label{rho_A map}
     \rho_A \; \stackrel{x}{\longmapsto} \rho_{A,x} \propto \; \uruleL (P_x^A, \rho_A)\,, 
\end{equation} 
where the operators $\rho_A=\text{Tr}_B(\rho_{AB})$ and $P_x^A$ are defined on the \textit{same} Hilbert space $\Hs_A$, as required by the \Luders rule. The relation \eqref{rho_A map}  does, of course, not mean that the system as a whole now resides in a product state.
Importantly, when obtaining a specific outcome in a single run, Eq.\ \eqref{rho_A map} says nothing about the post-measurement state of the other subsystem, or of the system as a whole.\footnote{It is only possible to conclude that the pre-measurement state $\rho_{AB}$ assigns a non-zero probability to obtaining the outcome $P^A_{x}$.} Read in this way, the statement of \Luders rule given by ($\mathsf{L}$) is \textit{incomplete} as it does not fully prescribe how to update the state of a quantum system when a local measurement is performed.\footnote{The update rule also seems to be incomplete if one were to use a different formulation of ($\mathsf{L}$) stating that, after a measurement, the system resides in an eigenstate of the measured observable.}

There are, however, consistency requirements when performing measurements on a composite system. If the outcome $P_x^A$ is \textit{non-degenerate}, then the \Luders  rule assigns a \textit{pure} post-measurement state $\uruleL (P_x^A, \rho_A)$ to subsystem $A$. Consequently,  the post-measurement state of the composite system must be a \textit{product} state, \textcolor{modification}{regardless of} the initial state  $\rho_{AB}$. On its own, the projection postulate ($\mathsf{L}$) provides, however, \textit{no} information about the local state at $B$ upon performing a local measurement of $M_A$ on the constituent $A$. In principle,\textit{ any} density matrix $\rho_B  = \textcolor{modification}{\text{Tr}_A} (\rho_{AB})  \in \mathcal{S}(\Hs_B)$ is compatible.  For example, one could imagine that any local measurement at $A$  throws the state of the composite system into a product state, with the post-measurement state at $B$ being equal to the maximally mixed state $\rho_B \propto \id_B$.  

The correct global post-measurement state after a local measurement on a quantum system is, of course, well-known. Calculating quantum correlations in Bell-type experiments relies on knowing the state of subsystem $B$ if the outcome of a local measurement at $A$ has been found.\footnote{Recall that, mathematically, the state that correctly reproduces the experimental findings is obtained in the following way. Write the initial state corresponding to $\rho_{AB}$ as a superposition of product terms, using the  eigenstates of the operator $M_A$ for the first factor. Then, identify the term with the label $x$, the value of the observed measurement outcome in a given run. The second factor in this expression characterises the correct post-measurement state for subsystem $B$.} Formally, upon measuring the local observable $M_A$ and finding $P_x^A$ given $\rho_{AB} \in \mathcal{S}(\Hs_A \otimes \Hs_B)$, the state of the composite system updates in this way,
\begin{equation} \label{eq: normalised collapsed bipartite state}
       \rho_{AB} \; \stackrel{x}{\longmapsto} \; \rho_{AB,x} = \frac{P_x^A \otimes \id_B \, \rho_{AB} \, P_x^A\otimes \id_B}{\text{Tr} \left[ P_x^A \otimes \id_B \, \rho_{AB} \, P_x^A \otimes \id_B \right] } \in \mathcal{S}(\Hs_A \otimes \Hs_B)\, . 
    \end{equation}

This rule generalises the postulate ($\mathsf{L}$) to cover the update of the \textit{global} state $\rho_{AB}$ caused by a \textit{local} measurement,
\begin{equation} \label{eq: normalised collapsed bipartite state 2}
       \rho_{AB} \; \stackrel{x}{\longmapsto} \; \uruleL(P_x^A,\rho_{AB}) = P_x^A \otimes \id_B \, \rho_{AB} \, P_x^A\otimes \id_B \in \bar{\mathcal{S}}(\Hs_A \otimes \Hs_B)\, .
    \end{equation}
    
Stated differently, we are making the point that the projection postulate ($\mathsf{L}$) or, equivalently, L{\"u}ders  rule \eqref{Lueders}, does \textit{not} say how the state of a composite system updates under local measurements. Additional \textit{experimental} input is required to arrive at the rule \eqref{eq: normalised collapsed bipartite state}. For simplicity, we continue to use the symbol $\uruleL$ to denote the \textit{complete} update rule which defines the effects of both local \textit{and} global measurements.  

The update rule \eqref{eq: normalised collapsed bipartite state} does have a simple  physical interpretation. Suppose we use a device $\mathcal{D}_{AB}$ to perform a \textit{global} measurement of the observable  $M_A \otimes \id_B$, with  outcomes represented by the projector $P_x^A \otimes \id_B$. To find the post-measurement state of the composite system, we can apply the original form of the projection postulate ($\mathsf{L}$), only to find the update rule  \eqref{eq: normalised collapsed bipartite state}. 

The distinction between a local measurement of $M_A$ implemented by the device $\mathcal{D}_A$ and a global one of $M_A \otimes \id_B$, effected by the device $\mathcal{D}_{AB}$, will be important to express the requirements that alternatives to the \Luders rule must satisfy.

\section{Generalising the state-update rule} \label{AURules}

\subsection{Preliminaries}

In this section, we will define the properties that an acceptable state-update rule must have. The main role of an update rule is, of course, to determine the post-measurement state of the measured system. We need to allow that measurements on a part of a composite system may affect the state of other subsystems not acted upon. This behaviour is known to occur in quantum theory and gave rise to the discussion in Sec.\ \ref{LuedersOnCompSys}. In other words, we must ensure that an update rule assigns post-measurement states consistently to \textit{composite} systems when local measurements are performed. For convenience, we will focus on bipartite systems; the extension to multipartite systems is straightforward.\footnote{Anticipating the discussion of Sec.\ \ref{sec: proper vs improper} on the role \textit{proper} and \textit{improper} mixed states, we assume for now that the density operators we consider represent either pure states or improper mixed states, i.e.\ those arising from entanglement with other systems.}

A generic update rule will be denoted by the letter $\omega$, with a subscript identifying the possibly composite system to which it applies. Let us consider a composite system with Hilbert space $\Hs_A \otimes \Hs_B$, initially residing in the state $\rho_{AB} \in \mathcal{S}(\Hs_A \otimes \Hs_B)$. If a local measurement is performed on subsystem $\Hs_A$ and yields the outcome $P_x^A$, the operator $\uruleA_{AB}(P_x^A, \rho_{AB})$ will represent the sub-normalised post-measurement state of the system, in analogy with the \Luders rule of Eq.\ \eqref{eq: normalised collapsed bipartite state 2}.\footnote{A different formula is generally required to describe the update of \textit{proper} mixed states, i.e.\ those due to the incomplete knowledge about the preparation of a system (cf.\ Sec.\ \ref{sec: proper vs improper}).} The subsystem on which the measurement is performed can be deduced from the Hilbert space associated with the projector $P_x^A$ characterising the measurement outcome. In the current example, it is the system with label $A$.

Any update rule must conform with postulate ($\mathsf{B}$) implying a non-trivial constraint. By taking the trace of a sub-normalised post-measurement state, the Born rule prescribes the probability to obtain outcome $P_x^A$,
\begin{equation} \label{eq: Born rule consistency}
        \text{Tr}\left[\uruleA_{AB}\left( P_x^A, \rho_{AB} \right) \right] = \text{Tr}\left( P_x^A \otimes \id_B  \, \rho_{AB} \right) = \text{Tr}\left( P_x^A   \, \text{Tr}_B \left( \rho_{AB} \right) \right) \, .
\end{equation}
An analogous condition is, of course, assumed for measurements on subsystem $\Hs_B$.

When considering sequences of measurements, \textit{1-homogeneity}, described following Eq.\ \eqref{Lueders} for the quantum-mechanical update rule, is a convenient mathematical property. It means that an update rule satisfies the constraint
\begin{equation} \label{eq: 1-homogeneity}
    \uruleA_{AB}  \left( P_x^A, \lambda \, \rho_{AB} \right) = \lambda \, \uruleA_{AB} \left( P_x^A, \rho_{AB} \right)\,, \quad \lambda \in \left[0,1 \right] \,,
\end{equation}
for all initial states $\rho_{AB} \in \bar{\mathcal{S}}(\Hs_A \otimes \Hs_B)$  and measurement outcomes $P_x^A\in\mathcal{P}(\Hs_A)$. This property allows us to concatenate update rules without the need to repeatedly normalise post-measurement states. Now suppose that two time-ordered measurements with outcomes $P_x$ and $P_y$ are performed on a composite system in state $\rho_{AB}$. Then, the final state is simply given by $\uruleA_{AB} \left( P_y, \uruleA_{AB} \left( P_x, \rho_{AB} \right) \right)$. Here the outcomes $P_x$ and $P_y$ may correspond to measurements performed on the same subsystem or on different subsystems.

\subsection{Necessary properties of generalised update rules} \label{subsec: assumptions we make}

We will now list physically or operationally motivated requirements for update rules. These assumptions will lead to a concise definition of a generalised update rule, given in Def.\ \ref{def: update rule}. Naturally, these requirements are satisfied by the quantum mechanical \Luders rule.

The requirements will take into account that update rules must hold for single and composite systems. For simplicity, we will usually consider only measurements on one subsystem, say $\Hs_A$, though analogous conditions are assumed to apply to measurements on \emph{any} subsystem.

The first three assumptions do not explicitly rely on the composite structure.
\textit{Completeness} requires that every conceivable measurement outcome specifies a unique post-measurement state. \textit{Context-independence} asserts that the observed outcome and the pre-measurement state alone determine the post-measurement state. It introduces a form of non-contextuality \cite{kochen_problem_1967} that is present in quantum theory---namely, that all operationally equivalent experimental runs (i.e.\ involving the same initial state and the same measurement outcome) yield identical post-measurement states. Mathematically, completeness and context-independence determine the co-domain and the domain of the update rule $\uruleA_{AB}$ for a composite system, respectively.
\begin{description}
    \item[\textbf{(A1-A2)}] \textbf{Completeness} $\&$ \textbf{context-independence} A complete and context-independent update rule maps any combination of a (local) outcome, $P_x^A\in\mathcal{P}(\Hs_A)$, and a sub-normalised state, $\rho_{AB} \in \bar{\mathcal{S}}(\Hs_A \otimes \Hs_B)$, to another sub-normalised state,
    \begin{equation} \label{eq: completeness & context-independence}
    \uruleA_{AB}: \mathcal{P}(\Hs_A) \times \bar{\mathcal{S}}(\Hs_A \otimes \Hs_B) \to \bar{\mathcal{S}}(\Hs_A \otimes \Hs_B) \, .
\end{equation}
\end{description}

Thirdly, \textit{local covariance} ensures that different observers provide consistent descriptions of the same experiment. For a measurement on a non-composite system, suppose Alice's description involves the state $\rho$ and outcome $P_x$. In Bob's reference frame, the same measurement will be described by the state $U \rho U^{\dagger}$ and the outcome $U P_x U^{\dagger}$, where $U$ represents the frame change. Local covariance ensures that Alice's and Bob's post-measurement states remain related by $U$, preventing any contradictory predictions concerning future measurements. Effectively, we require that, instead of applying a unitary $U$ on a system after measuring an observable $\{ P_x \}_x$, one can equivalently implement it before measuring the suitably transformed observable $\{ UP_x U^{\dagger} \}_x$. The term ``local'' gains particular significance when considering the more general scenario of a measurement performed on a subsystem of a composite system. To ensure consistency across descriptions of a local measurement, the same joint state must be obtained when applying the unitary transformation $U_A \otimes U_B$ (\textit{i}) after measuring $\Hs_A$ with outcome $P_x^A$, or (\textit{ii}) before measuring $\Hs_A$ with rotated outcome $U_A P_x^A U_A^{\dagger}$. Local unitaries $U_A \otimes U_B$ represent changes to frames that preserve the composite structure of $\Hs_A \otimes \Hs_B$.
\begin{description}
    \item[\textbf{(A3)}] \textbf{Local covariance} For all measurement outcomes $P_x^A\in\mathcal{P}(\Hs_A)$, states $\rho_{AB} \in \bar{\mathcal{S}}(\Hs_A \otimes \Hs_B)$ of the joint system, and local transformations $U_A \in \mathcal{U}(\Hs_A)$ and $U_B \in \mathcal{U}(\Hs_B)$, we must have
    \begin{equation} \label{eq: unitary invariance}
        \uruleA_{AB} \left( U_A P_x^AU_A^{\dagger}, U_A\otimes U_B \rho_{AB} U_A^{\dagger} \otimes U_B^{\dagger} \right) =\left(U_A \otimes U_B  \right) \left(\uruleA_{AB} \left( P_x^A, \rho_{AB} \right)  \right) \left(U_A^{\dagger} \otimes U_B^{\dagger}  \right) \, .
    \end{equation}
\end{description}

The remaining assumptions rely on the composite structure of the system.
We require \textit{self-consistency} of the update rule: the assigned post-measurement states must not depend on whether or not the probed system is regarded as part of a larger composite system. This property ensures that the assignment of post-measurement states for composite systems is unambiguous. Violating self-consistency would result in an inconsistent treatment of measurements in the framework, wherein an observer's choice of whether to regard a system as a \textit{sub}system (and how to partition its environment) would have measurable effects. Mathematically, self-consistency is ensured by the requirement that the state of any system after the measurement could also be obtained by applying $\uruleA$ to a larger system and tracing out the irrelevant subsystems.
\begin{description}
    \item[\textbf{(A4)}] \textbf{Self-consistency} Suppose that the system with label $B$ is, in fact, a composite system, i.e.\ $\Hs_B = \Hs_{B^{\prime}} \otimes \Hs_{B^{\prime \prime}}$, hence the update rule can be relabeled as $\uruleA_{AB}=\uruleA_{AB^{\prime}B^{\prime \prime}}$. For all $\Hs_{B^{\prime \prime}}$, joint states $\rho_{AB^{\prime}B^{\prime \prime}} \in \bar{\mathcal{S}}(\Hs_{A} \otimes \Hs_{B^{\prime}} \otimes \Hs_{B^{\prime \prime}})$ and measurement outcomes $P_x \in \mathcal{P}(\Hs_A)$, we must have
    \begin{equation} \label{eq: self-consistency}
        \text{Tr}_{B^{\prime \prime}} \left[ \uruleA_{AB^{\prime}B^{\prime \prime}} \left( P_x^A, \rho_{AB^{\prime}B^{\prime \prime}} \right) \right] = \uruleA_{AB^{\prime}} \left( P_x^A, \text{Tr}_{B^{\prime \prime}} \left( \rho_{AB^{\prime}B^{\prime \prime}} \right) \right)  \, .
    \end{equation}
\end{description}
Self-consistency entails that the update rule $\uruleA_A$ of a non-composite system emerges as a special case of the composite-system update rule $\uruleA_{AB}$: letting $\rho_A = \text{Tr}_B (\rho_{AB})$, we have
\begin{equation} \label{eq: single rule from composite rule}
    \uruleA_A \left( P_x^A, \rho_A \right) = \text{Tr}_B \left( \uruleA_{AB} \left( P_x^A , \rho_{AB} \right) \right) \, .
\end{equation}

Any update rule must respect the (quantum) \textit{no-signalling principle}, which prohibits space-like separated parties from communicating through \textit{local} measurements. Mathematically, we require that unconditional measurements---where no specific outcome is selected---do not alter the reduced state of unmeasured subsystems. 
\begin{description}
    \item[\textbf{(A5)}] \textbf{No-signalling} For all joint states $\rho_{AB} \in \bar{\mathcal{S}}(\Hs_{A} \otimes \Hs_B )$ and all local observables $\{ P_x^A \}_x$ on $\Hs_A$, we must have
    \begin{equation} \label{eq: no-signalling}
        \sum_x \text{Tr}_A \left[ \uruleA_{AB} \left( P_x^A, \rho_{AB} \right) \right] = \text{Tr}_A \left( \rho_{AB} \right) \, .
    \end{equation}
\end{description}

Finally, we assume \textit{local commutativity}: conditioned on fixed outcomes for measurements on different subsystems, the order in which these measurements are performed does not matter. As with no-signalling, local commutativity finds justification in special relativity. Suppose that local measurements on two space-like separated subsystems yield outcomes $P_x^A$ and $P_y^B$. Since the time-ordering of the two measurements is relative, we require that the same post-measurement state of the combined system is obtained irrespective of whether $P_x^A$ or $P_y^B$ is measured first. This property also implies that the post-measurement state of the combined system is defined unambiguously if $P_x^A$ and $P_y^B$ are measured simultaneously.
\begin{description}
    \item[(A6)] \textbf{Local commutativity} For all measurement outcomes $P_x^A\in \mathcal{P}(\Hs_A)$, $P_y^B\in \mathcal{P}(\Hs_B)$ and joint states $\rho_{AB} \in \bar{\mathcal{S}}(\Hs_{A} \otimes \Hs_B )$, we must have
    \begin{equation} \label{eq: local commutativity}
        \uruleA_{AB} \left( P_x^A, \uruleA_{AB} \left( P_y^B, \rho_{AB} \right) \right) = \uruleA_{AB} \left( P_y^B, \uruleA_{AB} \left( P_x^A, \rho_{AB} \right) \right) \, . 
    \end{equation}
\end{description}

Requirements A1-A6 are the core assumptions feeding into the framework of generalised state-update theories. Other features of the \Luders rule of quantum theory will not necessarily be shared by generalised update rules. For example, the resulting foil theories do not have to be deterministically repeatable, and they do not have to allow for quantum measurements to preserve the indistinguishability of preparations (cf.\ Sec.\ \ref{sec: proper vs improper}). As we will see in Sec.\ \ref{sec: examples}, there are generalised state-update theories (GUTs) in which different preparations of the same mixed state do not remain indistinguishable after a measurement. Thus, depending on the update rule, density operators may or may not provide a complete description of an individual system.

We are now in a position to present the definition of generalised state-update rules.

\begin{defn} \label{def: update rule}
   A \emph{generalised state-update rule} $\uruleA$ is a set of functions $\{ \uruleA_{AB} \}$---where each function $\uruleA_{AB}$ is associated with a pair of Hilbert spaces $\Hs_A$ and $\Hs_B$---satisfying the following properties: (\textit{i}) consistency with the Born rule (Eq.\ \eqref{eq: Born rule consistency}); (\textit{ii}) 1-homogeneity (Eq.\ \eqref{eq: 1-homogeneity}); and (\textit{iii}) assumptions A1-A6 (Eqs.\ \eqref{eq: completeness & context-independence}-\eqref{eq: local commutativity}), e.g.\ $\uruleA_{AB}: \mathcal{P}(\Hs_A) \times \bar{\mathcal{S}}(\Hs_A \otimes \Hs_B) \to \bar{\mathcal{S}}(\Hs_A \otimes \Hs_B)$ etc.
\end{defn}
Any update rule chosen to replace the projection postulate of quantum theory will give rise to one specific generalised state-update theory. Multiple update rules may induce the same behaviour on non-composite systems via Eq.\ \eqref{eq: single rule from composite rule}. This point is illustrated in Sec.\ \ref{Correlation-freeQM} where different extensions of the single-system \Luders rule to multi-partite systems are presented.

How do the update rules defined here compare with \textit{quantum instruments} \cite{heinosaari_mathematical_2011}, i.e.\ collections $\{\mathcal{I}_x\}_x$ of \textit{quantum operations}\footnote{Quantum operations are linear, completely positive and trace non-increasing maps defined on the space $\mathcal{L}(\Hs)$ of bounded operators acting on $\Hs$. Quantum instruments are usually defined as mappings from an outcome space $\left( X, \Sigma \right)$ to the set of quantum operations. However, since we will only deal with discrete observables, an instrument is completely determined by the finite set of operations $\{\mathcal{I}_x\}_x$.} that describe the effect on states of quantum measurements? The main difference is that the state-update rules satisfying our definition are \textit{not} required to be linear or completely positive. Consequently, the set of alternative state-update rules is larger than the set of quantum instruments.

\subsection{Updating (im-)proper mixtures} \label{sec: proper vs improper}

For a pure joint state $\rho_{AB}$, Eq.\ \eqref{eq: single rule from composite rule} applies irrespective of $\rho_{AB}$ being a product state or being entangled. Accordingly, $\uruleA_{A} \left( P_x^A, \rho_{A} \right)$ represents the post-measurement state of system $\Hs_A$ when $\rho_{A}$ is either a pure state or an \textit{improper} mixture, i.e.\ the reduced density operator of an entangled state.
In contrast, \textit{proper} mixtures result from classical uncertainty about the ``true'' preparation of a system. A proper mixture can be thought of as a \textit{Gemenge} (German for `mixture'), i.e.\ as a collection $\mathcal{G}=\{(p_1,\rho_1), (p_2, \rho_2),\dots, (p_n,\rho_n) \}$ where $p_i\geq 0$, $\sum_i p_i=1$ and $\rho_i \in \bar{\mathcal{S}}(\Hs)$ represent pure states or improper mixed states. The Gemenge $\mathcal{G}$ describes a system that is prepared in state $\rho_i$ with probability $p_i$ \cite{busch_quantum_1996}. To any Gemenge there corresponds a unique density operator, $\rho_{\mathcal{G}} = \sum_{i=1}^n p_i \rho_i$.

If a system described by $\mathcal{G}$ is measured with outcome $P_x^A$, the updated description is given by the Gemenge $\mathcal{G}_x$ with the same weights $p_i$ but with updated (sub-normalised) states, $\mathcal{G}_x = \{ p_i, \uruleA_A (P_x^A, \rho_i) \}$. In other words, if the system resided in the state $\rho_i$ with probability $p_i$ before the measurement with outcome $P_x^A$, it will reside in the state $\uruleA_A (P_x^A, \rho_i)$ with probability $p_i$ after the measurement.\footnote{If the outcome $P_x^A$ is never observed, $\text{Tr}(P_x^A \rho_i)=0$, then $\uruleA_A (P_x^A, \rho_i)=O$, where $O$ is the zero operator, and the experimenter learns that the system was \textit{not} prepared in the state $\rho_i$.} Thus, the density matrix updates according to the rule
\begin{equation} \label{eq: update rule on proper mixture}
    \rho_{\mathcal{G}} = \sum_{i=1}^n p_i \rho_i \quad \stackrel{P_x^A}{\loongmapsto} \quad \rho_{\mathcal{G}_x} = \sum_{i=1}^n p_i \; \uruleA_{A} \left( P_x^A, \rho_i \right) \, .
\end{equation}
A similar expression applies for local measurements on a subsystem when the composite system is described by a proper mixture, i.e.\ $\mathcal{G}=\{(p_1,\rho^1_{AB}), (p_2, \rho^2_{AB}),\dots, (p_n,\rho^n_{AB}) \}$.

The update rule is therefore applied differently depending on whether a mixed state is regarded as a proper or an improper mixture. As shown in \eqref{eq: update rule on proper mixture}, in the case of a proper mixture, the map acts on each operator $\rho_i$ appearing in the decomposition of $\rho_{\mathcal{G}}$ and included in the associated Gemenge $\mathcal{G}$. Typically, the post-measurement state $\rho_{\mathcal{G}_x}$ also represents a proper mixture. Consequently, sequences of measurements on a proper mixture are described by repeatedly applying the update rule to each pure or improper mixed state $\rho_i$ in the initial Gemenge $\mathcal{G}$. 

The equations \eqref{eq: completeness & context-independence}-\eqref{eq: local commutativity} formalising assumptions A1-6 were formulated with pure and improper mixed states in mind. Nevertheless, they also ensure a consistent treatment of measurements on proper mixtures. For instance, the completeness A1 of $\uruleA$ guarantees that $\rho_{\mathcal{G}_x}$ in \eqref{eq: update rule on proper mixture} is a valid state of the system. Similar arguments apply to assumptions A3-6, but not to context-independence A2, which requires minor adaptation. In fact, the post-measurement state $\rho_{\mathcal{G}_x}$ generally does not depend on the initial mixed state $\rho_{\mathcal{G}}$, but rather on the initial Gemenge $\mathcal{G}$. When proper mixtures are included, context-independence is generalised as follows: the observed outcome and the pre-measurement Gemenge alone determine the post-measurement state.\footnote{If $\rho_A$ denotes either a pure or an improper mixed state, we can associate the trivial Gemenge $\mathcal{G} = \{ (1, \rho_A) \}$, describing a system prepared in the state $\rho_A$ with probability $1$.}

By considering proper and improper mixtures separately, we allow the framework to include foil theories where sequential measurements can be used to distinguish between different realisations of the same mixed state. Unlike quantum theory, such theories do not generally preserve the indistinguishability of preparations. In addition, our framework should also include update rules that send pure states to proper mixtures, thus introducing another level of uncertainty into the behaviour of a measuring device. In such cases, equations such as \eqref{eq: local commutativity} for local commutativity A6, which involve sequential measurements, must be adjusted by ensuring that the second update rule is applied to each state in the Gemenge resulting from the first measurement.

\subsection{Foil theories from generalised state-update rules} \label{sec: Sec with def GUTs}

Having defined update rules in Def.\ \ref{def: update rule}, we now introduce GUTs.

\begin{defn} \label{def: AMT}
A \emph{Generalised state-Update Theory} (GUT) is defined by the quantum postulates for states ($\mathsf{S}$), time evolution ($\mathsf{T}$), system composition ($\mathsf{C}$), observables ($\mathsf{O}$), outcome probabilities ($\mathsf{B}$), and by a generalised state-update rule $\uruleA$ assigning post-measurement states. 
\end{defn}

\section{Examples of generalised state-update rules} \label{sec: examples}

We now introduce several examples of update rules and investigate properties they have in addition to the essential requirements A1-A6 of Sec.\ \ref{subsec: assumptions we make}. Each of the hypothetical rules defines a foil of quantum theory through Def.\ \ref{def: AMT} in Sec.\ \ref{sec: Sec with def GUTs}, characterised by a modified projection postulate.

\subsection{\Luders measurements: quantum theory} \label{sec: quantum theory (extensions)}
Our framework includes quantum theory, since the properties A1-A6 underlying Def.\ \ref{def: update rule} for generalised update rules have been abstracted from the \Luders rule,
\begin{equation} \label{eq: Lueders rule full}
    \uruleL_{AB}(P_x^A, \rho_{AB}) = P_x^A \otimes \id_B \rho_{AB} P_x^A \otimes \id_B \, .
\end{equation}

In Sec.\ \ref{LuedersOnCompSys}, we described an important property of the \Luders rule, that we will call \textit{composition compatibility}. A \textit{local} measurement on subsystem $\Hs_A$, implemented by a device $\mathcal{D}_A$ and yielding the outcome $P_x^A$, leads to the same post-measurement state of the composite system as a \textit{global} measurement on $\Hs_A \otimes \Hs_B$, implemented by a device $\mathcal{D}_{AB}$ and yielding the outcome represented by $P_x^A \otimes \id_B$.

\begin{defn} [\emph{Composition Compatibility}] \label{def: composition compatibility}
   An update rule $\uruleA$ satisfies \emph{composition compatibility} if, for every $\Hs_A$, $\Hs_B$, joint state $\rho_{AB}\in \bar{\mathcal{S}}(\Hs_A \otimes \Hs_B)$ and local outcome $P_x^A \in \mathcal{P}(\Hs_A)$,
   \begin{equation} \label{eq: def CC}
       \uruleA_{AB} \left( P_x^A, \rho_{AB} \right) = \uruleA_{AB} \left( P_x^A \otimes \id_B, \rho_{AB} \right) \, .
   \end{equation}
\end{defn}
We will see that not all update rules satisfy this property---in other words, it does not follow from assumptions A1-A6. Note that the two expressions in Eq.\ \eqref{eq: def CC} capture two different situations: the function on the left-hand side is defined on $\mathcal{P}(\Hs_A) \times \bar{\mathcal{S}}(\Hs_A \otimes \Hs_B)$, while that on the right-hand side on $\mathcal{P}(\Hs_A \otimes \Hs_B) \times \bar{\mathcal{S}}(\Hs_A \otimes \Hs_B)$. In physical terms, one side of Eq.\ \eqref{eq: def CC} corresponds to a \textit{local} measurement on $\Hs_A$, while the other side is associated with a \textit{global} measurement on $\Hs_A \otimes \Hs_B$. Eq.\ \eqref{eq: def CC} refers to measurements on $\Hs_A$ only, but the property is assumed to hold for measurements on any subsystem.

Composition compatibility \eqref{eq: def CC} can therefore be regarded as the additional assumption that extends the ``operationally incomplete'' postulate ($\mathsf{L}$) for single systems, i.e.\
\begin{equation} \label{eq: Luders for single systems}
    \uruleL_A(P_x^A, \rho_A)=P_x^A \rho_A P_x^A \, ,
\end{equation}
to the \Luders rule $\uruleL_{AB}$ for multi-partite systems, Eq.\ \eqref{eq: Lueders rule full}, which encompasses the effects of both local and global measurements.

Another argument to extend the single-system \Luders rule to composite systems uses the fact that $\uruleL_{A,x}(\cdot) \equiv \uruleL_{A}(P^A_x,\cdot): \bar{\mathcal{S}}(\Hs_A) \to \bar{\mathcal{S}}(\Hs_A)$, which denotes the \Luders rule for system $\Hs_A$ conditioned on the measurement outcome $P_x^A$, is convex-linear over $\bar{\mathcal{S}}(\Hs_A)$. Therefore, $\uruleL_{A,x}$ has a unique linear extension to the set of bounded operators $\mathcal{L}(\Hs_A)$ (cf.\ \cite{heinosaari_mathematical_2011}). The \Luders rule \eqref{eq: Lueders rule full} for the composite system $\Hs_A \otimes \Hs_B$, conditioned on the local measurement outcome $P_x^A$, is then obtained by setting 
\begin{equation} \label{eq: extension Luders (CP)}
    \uruleL_{AB} \left( P_x^A, \rho_{AB} \right) = \left(\uruleL_{A,x} \otimes \mathcal{I}_B \right) \left( \rho_{AB} \right) \,
\end{equation}
for all $\rho_{AB} \in \bar{\mathcal{S}}(\Hs_A \otimes \Hs_B)$, where $\left(\uruleL_{A,x} \otimes \mathcal{I}_B \right) \left( L_A \otimes L_B \right) = \uruleL_{A,x} \left( L_A  \right) \otimes L_B$ for all $L_A \in \mathcal{L}(\Hs_A)$ and $L_B \in \mathcal{L}(\Hs_B)$. Here, $\mathcal{I}_B$ denotes the identity channel on $\Hs_B$. Eq.\ \eqref{eq: extension Luders (CP)} is well-defined, as well as linear over $\bar{\mathcal{S}}(\Hs_A \otimes \Hs_B)$, because $\uruleL_{A,x}$ can be linearly extended to $\mathcal{L}(\Hs_A)$. In particular, the completeness A1 of $\uruleL_{AB} \left( P_x^A, \rho_{AB} \right)$ follows from the \textit{complete positivity} of $\uruleL_{A,x}$: if $\uruleL_{A,x}$ was not completely positive, then the operator $\uruleL_{AB} \left( P_x^A, \rho_{AB} \right)$ would not be positive for some $\rho_{AB}$.

Typically, generalised update rules have to neither satisfy deterministic repeatability nor composition compatibility. They may also lack some other properties such as \textit{preparation indistinguishability}, which ensures that different preparations of the same mixed state remain indistinguishable after measurements. This property is equivalent to the statement that ``density operators provide complete descriptions of quantum systems''. Moreover, \textit{ideality} (measurement leaves a state unchanged if the outcome is certain) and \textit{local tomography} (local measurements of $M_A$ and $M_B$ on different subsystems are operationally equivalent to a global measurement of the product observable $M_A \otimes M_B$ if the experimenters are allowed to communicate classically) may or may not be present. The \Luders rule is also endowed with a form of \textit{coherence}, according to which the outcome probability distribution of an observable $M$ is not affected by a prior coarse-grained measurement of $M$ (see Def.\ \ref{def: coherence} in Sec.\ \ref{sec: singling out Lueders}). Table \ref{table: properties of valid update rules} summarises the extent to which these and other properties hold for the hypothetical update rules introduced in the following sections. 

\begin{table}[ht]
    \begin{center}
        \begin{tabular}{lccccc}
        
           & $\uruleL$ & $\uruleA^{\mathsf{locL}}$ & $\uruleP$ & $\uruleA^{\mathsf{dep}}$ & $\uruleA^{\lambda}$ \\ \hline
          deterministic repeatability & \checkmark & \checkmark & $\times $ & $\times$ & $\times$ \\ \hline
          preparation indistinguishability & \checkmark & (\checkmark) & $\times$& \checkmark & $\times$ \\ \hline
          composition compatibility & \checkmark & $\times$ & \checkmark & $\times$& $\times$ \\ \hline
          ideality & \checkmark & (\checkmark) & \checkmark & $\times$& (\checkmark) \\ \hline
          local tomography & \checkmark & $\times$& $\times$& \checkmark & $\times$\\ \hline
          coherence & \checkmark & \checkmark & $\times$& $\times$& $\times$\\ \hline
          nonlocality & \checkmark & $\times$& $\times$& \checkmark & $\times$\\ \hline
          complete positivity & \checkmark & $\times$& $\times$& \checkmark & $\times$\\ \hline
          weak repeatability & \checkmark & \checkmark & $\times$& $\times$& \checkmark \\ \hline
        \end{tabular}
    \end{center}
    \caption{Summary of the operational properties of the update rules discussed in Sec.\ \ref{sec: examples}. Ticks in parentheses indicate that the property holds for measurements on \textit{non-composite} systems only.}
    \label{table: properties of valid update rules}
\end{table}

\subsection{Locally-\Luders measurements: correlation-free quantum theory} \label{Correlation-freeQM}
As noted in Sec.\ \ref{AURules}, the \Luders rule for single systems does not fix the state-update rule in multi-partite systems. Consequently, there will exist other update rules for composite systems that reduce to the single-system \Luders rule.
One example is the 1-homogeneous map
\begin{equation} \label{eq: correlation-free rule}
    \urule_{AB}^{\mathsf{locL}}\left(P_x^A,\rho_{AB}\right)=\frac{\uruleL_{A}\left(P_x^A, \text{Tr}_{B} \left(\rho_{AB}\right)\right) \otimes \text{Tr}_{A}\left(\rho_{AB}\right)}{\text{Tr}\left(\rho_{AB}\right)}
\end{equation}
which satisfies assumptions A1-A6 and thus defines a valid update rule. While manifestly different from the standard \Luders rule of Eq.\ \eqref{eq: Lueders rule full}, it assigns the same post-measurement states as postulate ($\mathsf{L}$) to non-composite systems. Using Eq.\ \eqref{eq: single rule from composite rule} and letting $\rho_A = \text{Tr}_B (\rho_{AB})$, we obtain
\begin{equation}
    \urule_{A}^{\mathsf{locL}}\left(P_x^A,\rho_{A}\right) = \text{Tr}_B \left( \urule_{AB}^{\mathsf{locL}} \left( P_x^A , \rho_{AB} \right) \right) = \urule_{A}^{\mathsf{L}}\left(P_x^A,\rho_{A}\right) \, .
\end{equation}
In other words, $\urule^{\mathsf{locL}}$ represents an alternative extension of postulate ($\mathsf{L}$) but, unlike the \Luders rule $\uruleL$, it does \textit{not} satisfy composition compatibility, and is \textit{not} convex-linear over pre-measurement states of the \textit{composite} system.

The GUT defined by the update rule \eqref{eq: correlation-free rule} is therefore \textit{locally indistinguishable} from quantum mechanics but creates no correlations between entangled systems. Regardless of the observed outcome, a local measurement on subsystem $A$ leaves the reduced state of subsystem $B$ unchanged. Although measurement breaks the entanglement between two systems---by always mapping $\rho_{AB}$ to a product state---distant parties will observe no correlation between their respective measurement outcomes. In such \textit{correlation-free} variant of quantum theory, Bell's inequalities cannot be violated, and entanglement is not available as a resource for distributed tasks.

\subsection{Non-collapsing measurements: passive quantum theory}  \label{sec: pQT}

We will now consider measurements that, instead of updating the quantum state, leave it unchanged.\footnote{The possibility of non-collapsing measurements was listed by von Neumann \cite{vonneumann_mathematische_1932} as early as 1932 as one of three possible reactions of a physical system to a measurement.} Outcomes still occur probabilistically as dictated by the Born rule, yet, in analogy with classical theory, these measurements do not disturb the pre-measurement state. This situation corresponds to a ``\textit{passive}'' update rule $\uruleP_A$ for non-composite systems,
\begin{equation} \label{eq: passive rule single systems} \uruleP_A \left( P_x^A, \rho_A \right) = \text{Tr} \left( P_x^A \rho_A \right) \rho_A \,. 
\end{equation}
Non-collapsing measurements, originally introduced in \cite{aaronson_space_2016}, cannot be implemented in quantum theory \cite{busch_is_1997}. Indeed, $\uruleP_A$ is not convex-linear over $\bar{\mathcal{S}}(\Hs_A)$ which implies that, after a passive measurement, one may be able to distinguish between different preparations of the same mixed state. For example, let $\Hs_A=\mathbb{C}^2$ and consider the Gemenge 
\begin{equation}
\mathcal{G}=\{(1/2, \kb{0}{0}), (1/2, \kb{1}{1})\} \quad \mbox{and} \quad \mathcal{G}^{\prime}=\{(1/2, \kb{+}{+}), (1/2, \kb{-}{-})\}\,,
\end{equation} where $\ket{\pm}=(\ket{0} \pm \ket{1})/\sqrt{2}$. Both preparations correspond to the maximally mixed state, $\rho_{\mathcal{G}}=\rho_{\mathcal{G}^{\prime}}= \id_2/2$. According to the Gemenge-update rule \eqref{eq: update rule on proper mixture} in Sec.\ \ref{sec: proper vs improper}, a passive measurement with outcome $P_x^A=\kb{0}{0}$ will produce two different sub-normalised density operators,
\begin{equation}
    \rho_{\mathcal{G}_x} = \frac{1}{2} \kb{0}{0} \quad \neq \quad \rho_{\mathcal{G}^{\prime}_x} = \frac{\id_2}{4} \, .
\end{equation}
Therefore, the non-collapsing update rule $\uruleP_A$ does not satisfy preparation indistinguishability, indicating that density operators do not provide complete descriptions of individual systems.

While complete positivity can be defined for both linear and nonlinear transformations \cite{ando_non-linear_1986}, it has been described as ``physically unfitting'' for nonlinear dynamics \cite{czachor_complete_1998}. In fact, although $\uruleP_A$ (or, more precisely, extensions of $\uruleP_A$ to $\mathcal{L}(\Hs_A)$) is not completely positive, the update rule can still be consistently extended to composite systems. Invoking compositional compatibility (Def.\ \ref{def: composition compatibility}), we set
\begin{equation} \label{eq: passive rule composite systems}
        \uruleP_{AB} \left( P_x^A, \rho_{AB} \right) = \uruleP_{AB} \left( P_x^A \otimes \id_B, \rho_{AB} \right) = \text{Tr} \left( P_x^A \otimes \id_B \rho_{AB} \right) \rho_{AB} \, .
\end{equation}
The map $\uruleP$ satisfies all the conditions in Def.\ \ref{def: update rule} and defines a GUT in which measurements do not modify the state of either the measured system or of the larger composite system to which it belongs. We called this theory \textit{passive quantum theory} (pQT).
An overview of pQT and its key properties has been provided in \cite{fiorentino_quantum_2023}. 

The absence of any measurement-induced state update allows, in principle, for tomographic reconstruction of the state of an individual system. Thus, the cloning of arbitrary states becomes possible and makes computation in pQT more efficient (in some sense) than in standard quantum theory \cite{busch_is_1997, kent_nonlinearity_2005}. The passive update rule of pQT also satisfies ideality but is not locally tomographic and does not allow for violations of Bell's inequalities (cf.\ Table \ref{table: properties of valid update rules}).

\subsection{Depolarising measurements} \label{sec: depolarising measurements}
Next, consider a theory in which measurements disturb the system to the point that the post-measure\-ment state contains no information about the original state. Repeating measurements on the same system, therefore, extracts no additional information about it. 

Updating the pre-measurement state to the maximally mixed state, regardless of the observed outcome, provides one possible realisation of such a theory,
\begin{equation} \label{eq: depolarising rule single systems}
    \urule^{\mathsf{dep}}_A \left( P_x^A, \rho_A \right) = \text{Tr} \left( P_x^A \rho_A  \right) \frac{\id_{A}}{d_A} \, ,
\end{equation}
where $d_A= \text{dim} (\Hs_A)$. The map in Eq.\ \eqref{eq: depolarising rule single systems} is both convex-linear over $\bar{\mathcal{S}}(\Hs_A)$ and completely positive. The set $\{ \urule^{\mathsf{dep}}_A (P_x^A) \}_x$ represents, in fact, a quantum instrument\footnote{Specifically, it belongs to the class of \textit{trivial instruments}.} compatible with the observable $\{P_x^A \}_x$.

Following Eq.\ \eqref{eq: extension Luders (CP)}, by defining $\urule^{\mathsf{dep}}_{A,x}(\cdot) \equiv \urule^{\mathsf{dep}}_{A}(P^A_x,\cdot)$, we can extend $\urule^{\mathsf{dep}}_A$ to composite systems via
\begin{equation} \label{eq: depolarising rule composite systems}
    \urule^{\mathsf{dep}}_{AB} \left( P_x^A, \rho_{AB} \right) = \left( \urule^{\mathsf{dep}}_{A,x} \otimes \mathcal{I}_B \right) \left( \rho_{AB} \right) \, .
\end{equation}
In contrast to the \Luders rule, the \textit{depolarising rule} $\urule^{\mathsf{dep}}$ does not satisfy composition compatibility. To see this, consider $\rho_{AB}= \kb{00}{00} \in \mathcal{S}(\Hs_2 \otimes \Hs_2)$ and $P_x^A=\kb{0}{0}$. We find that
\begin{equation} \label{eq: depolarising rule does not satisfy CC}
    \urule^{\mathsf{dep}}_{AB} \left( \kb{0}{0}, \kb{00}{00} \right) = \frac{\id_2}{2} \otimes \kb{0}{0} \quad \neq \quad \frac{\id_2}{2} \otimes \frac{\id_2}{2} = \urule^{\mathsf{dep}}_{AB} \left( \kb{0}{0} \otimes \id_B, \kb{00}{00} \right) \, .
\end{equation}
This example also demonstrates how composition compatibility could not have been used to extend $\urule^{\mathsf{dep}}_A$ to composite systems. The right-hand side of Eq.\ \eqref{eq: depolarising rule does not satisfy CC} cannot represent the post-measurement state of the composite system following a local measurement on a single qubit, as it would violate the no-signalling condition A5. See Table \ref{table: properties of valid update rules} for further properties of depolarising measurements.

\subsection{Probability-amplifying projective measurements} \label{sec: partially repeatable measurements}
Suppose an observable is measured twice on the same system in quick succession. Being deterministically repeatable, quantum theory predicts that we will observe the same outcome for the second measurement with certainty. We can construct update rules that satisfy a weaker notion of repeatability, called \textit{weak repeatability}: the probability of observing an outcome for the second time is \textit{greater} than the probability of observing it for the first time. Consider the following one-parameter family of update rules for non-composite systems:
\begin{equation} \label{eq: PRTs update rule}
    \urule^{\lambda}_A \left( P_x^A, \rho_A \right) = \tr{\left(  P_x^A  \rho_A \right)} \frac{G_{\lambda}^{1/2} \left( P_x^A \right)\, \rho_A \, G_{\lambda}^{1/2} \left( P_x^A \right)}{\tr{\left( G_{\lambda} \left( P_x^A \right)\,  \rho_A \right)}} \, ,
\end{equation}
where
\begin{equation} \label{eq: operator in PRTs}
     G_{\lambda} \left( P_x^A \right) =  \left( 1 - \lambda \right) P_x^A + \lambda \id_A \, , \quad \lambda \in [0,1] \, .
\end{equation}
If a measurement outcome $P_X^A = \sum_{i \in X}\kb{i}{i}$ for $X \subseteq \{ 0,...,d_A-1 \}$ is obtained given the pure state $\ket{\psi} = \sum_{i=0}^{d_A-1} c_i \ket{i} \in \Hs_A$, the resulting post-measurement state reads
\begin{equation} \label{eq: action of PRT rule}
    \urule^{\lambda}_A \left( P_X^A, \kb{\psi}{\psi} \right) =  \tr{\left(  P_x^A  \rho_A \right)} \kb{\psi_X}{\psi_X} \, ,
\end{equation}
where 
\begin{equation} \label{eq: pure pm state PRTs}
    \ket{\psi_X} = \frac{1}{N} \left( \sum_{j\in X} c_j \ket{j} + \sqrt{\lambda} \sum_{i \notin X} c_i \ket{i} \right)\, ,
\end{equation}
with $N$ being the normalisation factor. Therefore, the update rule $\urule^{\lambda}_A$ \textit{projects} $\ket{\psi}$ to the pure state $\ket{\psi_X}$ and the parameter $\lambda$ determines how likely a second measurement will yield the same outcome as the first, with a lower value of $\lambda$ corresponding to a higher probability of repeating the first outcome. The \Luders rule $\uruleL_A$ and the passive rule $\uruleP_A$ for non-composite systems are recovered by setting $\lambda=0$ and $\lambda=1$, respectively. Thus, for non-composite systems, the update rule $\urule^{\lambda}_A$ indeed smoothly interpolates between quantum theory and pQT. Weak repeatability occurs for $\lambda \in [0,1)$. The map in Eq.\ \eqref{eq: PRTs update rule} is not convex-linear over $\bar{\mathcal{S}}(\Hs_A)$ for $\lambda\neq 0$, hence the corresponding theories will not satisfy preparation indistinguishability.

Following Eq.\ \eqref{eq: correlation-free rule}, we find that a valid extension of Eq.\ \eqref{eq: PRTs update rule} to composite systems is provided by
\begin{equation} \label{eq: extension PRT rule}    \urule_{AB}^{\lambda}\left(P_x^A,\rho_{AB}\right)=\frac{\urule_{A}^{\lambda}\left(P_x^A, \text{Tr}_{B} \left(\rho_{AB}\right)\right) \otimes \text{Tr}_{A}\left(\rho_{AB}\right)}{\text{Tr}\left(\rho_{AB}\right)} \, .
\end{equation}
It can be shown that Eq.\ \eqref{eq: extension PRT rule} satisfies all conditions outlined in Def.\ \ref{def: update rule}. For each $\lambda\in (0,1)$, the \textit{probability-amplifying projective rule} $\urule^{\lambda}$ defines a \textit{weakly repeatable} GUT that is manifestly different from quantum theory.\footnote{Setting $\lambda=0$ recovers the locally-\Luders rule $\urule_{AB}^{\mathsf{locL}}$ of Eq.\ \eqref{eq: correlation-free rule}.} While alternative extensions of $\urule^{\lambda}_A$ may in principle be constructed, none would satisfy composition compatibility (Def.\ \ref{def: composition compatibility}). Similarly to the depolarising rule $\urule^{\mathsf{dep}}$, it is easy to show that the extension provided by Eq.\ \eqref{eq: def CC} violates no-signalling A5, as well as local commutativity A6.

\subsection{Invalid state-update rules} \label{sec: ruled out theories}

It is instructive to consider examples of state-update rules that \textit{violate}  at least one of the basic requirements A1-A6 of Def.\ \ref{def: update rule} in Sec.\ \ref{subsec: assumptions we make}. We have, in fact, already acknowledged two of them when presenting depolarising measurements (Sec.\ \ref{sec: depolarising measurements}) and probability-amplifying projective measurements (Sec.\ \ref{sec: partially repeatable measurements}). Extending the non-composite rules $\urule^{\mathsf{dep}}_A$ \eqref{eq: depolarising rule single systems} and $\urule^{\lambda}_A$ \eqref{eq: PRTs update rule} via compositional compatibility (Def.\ \ref{def: composition compatibility}) leads to maps assigning post-measurement states to composite systems in violation of the (quantum) no-signalling principle A5.

Table \textcolor{modification}{\ref{table: invalid update rules}} summarizes the properties of the update rules we will consider in this section.
\begin{table}[ht]
    \begin{center}
        \begin{tabular}{lccccc}
        
           & $\tilde{\uruleA}^{\mathsf{dep}}$ & $\tilde{\uruleA}^{\lambda}$ & $\uruleA^{\mu}$ & $\uruleA^{\mathsf{vN}}$ & $\uruleA^{\mathsf{U}}$ \\ \hline
          A1: completeness& \checkmark & \checkmark & \checkmark & \checkmark & \checkmark \\ \hline
          A2: context-independence& \checkmark & \checkmark & \checkmark & $\times$& \checkmark \\ \hline
          A3: local covariance& \checkmark & \checkmark & \checkmark & \checkmark & $\times$\\ \hline
          A4: self-consistency& \checkmark & \checkmark & \checkmark & \checkmark & \checkmark \\ \hline
          A5: no-signalling& $\times$& $\times$& \checkmark & \checkmark & \checkmark \\ \hline
          A6: local commutativity& \checkmark & $\times$& $\times$& \checkmark & \checkmark \\ \hline
        \end{tabular}
    \end{center}
    \caption{Summary of the state-updates described in Sec.\ \ref{sec: ruled out theories} that do not define an update rule as per Def.\ \ref{def: update rule}. For each rule, we record satisfied/violated assumptions using ticks and crosses, respectively. The first two, $\tilde{\uruleA}^{\mathsf{dep}}$ and $\tilde{\uruleA}^{\lambda}$, denote the extensions to composite systems of $\uruleA^{\mathsf{dep}}_A$ \eqref{eq: depolarising rule single systems} and $\uruleA^{\lambda}_A$ \eqref{eq: PRTs update rule}, respectively, obtained via composition compatibility (Def.\ \ref{def: composition compatibility}).}
    \label{table: invalid update rules}
    \end{table}

Consider now a convex mixture of the \Luders $\uruleL$ and the passive $\uruleP$ update rules for composite systems,
\begin{equation} \label{eq: convex mix luders and passive}
    \uruleA^{\mu}_{AB} \left( P_x^A, \rho_{AB} \right) = \left( 1- \mu \right) \uruleL_{AB} \left( P_x^A, \rho_{AB} \right) + \mu \, \uruleP_{AB} \left( P_x^A, \rho_{AB} \right), \quad \quad \mu \in (0,1)\, ,
\end{equation}
where the right-hand side describes a proper mixture. Eq.\ \eqref{eq: convex mix luders and passive} describes a local measurement that, with some \textit{a priori} probability $(1- \mu)$, collapses the state of the composite system as prescribed by quantum theory, or, with probability $\mu$, leaves the state unchanged. For $\mu \in (0,1)$, the probability of observing outcome $P^A_x$ in a possible second measurement is, on average, higher than in the first measurement. In other words, \eqref{eq: convex mix luders and passive} provides an alternative way to satisfy weak repeatability, different from the probability-amplifying projections $\urule^{\lambda}$ introduced in Sec.\ \ref{sec: partially repeatable measurements}.

The ``real'' post-measurement state of the system is either $\uruleL_{AB} \left( P_x^A, \rho_{AB} \right)$ or $\uruleP_{AB} \left( P_x^A, \rho_{AB} \right)$; one can imagine that this is determined by the value of some hidden variable $\lambda \in \{0,1\}$. The probability distribution of this hidden variable is the same across all experimental runs---i.e.\ $\text{prob}(\lambda=0)=1-\mu$, and $\text{prob}(\lambda=1)=\mu$---ensuring that it is independent of the context of a particular experiment. This guarantees that $\uruleA^{\mu}$ satisfies context-independence A2. Moreover, \eqref{eq: convex mix luders and passive} can be shown to satisfy completeness A1, local covariance A3, self-consistency A4 and no-signalling A5, since $\uruleA^{\mu}$ is a convex mixture of mappings that individually satisfy these conditions. Surprisingly, however, we find that \eqref{eq: convex mix luders and passive} violates local commutativity A6, despite both $\uruleL$ and $\uruleP$ satisfying it. 

The violation can be seen in a simple example. Let $\rho_{AB}=\kb{\Phi^+}{\Phi^+} \in \mathcal{S}(\Hs_2 \otimes \Hs_2)$, where $\ket{\Phi^+} = (\ket{00}+ \ket{11})/\sqrt{2}$ is a maximally entangled state, and consider $P_x^A= \kb{0}{0}$, and $P_y^B = \kb{1}{1}$. Recall that $\uruleA^{\mu}$ maps pure states to proper mixtures; hence, when representing sequential measurements, the second map must be applied to each pure state in the corresponding Gemenge, in line with the discussion in Sec.\ \ref{sec: proper vs improper}. If $P^A_x$ is obtained in a local measurement on $\Hs_A$ before the local measurement of $\Hs_B$ yields $P^B_y$, the state of the composite system updates as
\begin{equation}
    \kb{\Phi^+}{\Phi^+} \, \stackrel{(x \prec y)}{\loongmapsto} \, \frac{\mu}{4} \left( \left( 1- \mu \right) \kb{11}{11} + \mu \kb{\Phi^+}{\Phi^+} \right) \, .
\end{equation}
If instead the measurement yielding $P^B_y$ precedes the one yielding $P^A_x$, a different update occurs,
\begin{equation}
    \kb{\Phi^+}{\Phi^+} \, \stackrel{(y \prec x)}{\loongmapsto} \, \frac{\mu}{4} \left( \left( 1- \mu \right) \kb{00}{00} + \mu \kb{\Phi^+}{\Phi^+} \right) \, .
\end{equation}
Since the final state depends on the order in which the local measurements are performed, Eq.\ \eqref{eq: convex mix luders and passive} violates local commutativity A6. As a result, correlations between measurements on different subsystems cannot be unambiguously determined.

The projection postulate introduced by von Neumann \cite{vonneumann_mathematische_1932} is a notable example of state-update that does not define a valid update rule. For non-degenerate measurements, von Neumann's and L{\"u}ders' postulates coincide. However, according to von Neumann, outcome degeneracy of quantum measurements---represented by projectors of rank greater than one---arises from classical post-processing of non-degenerate measurements. That is, one does not directly implement a measurement of a degenerate observable $M$; instead, a measurement of a \emph{refinement} of $M$, i.e.\ a non-degenerate observable $M^{\prime}$ commuting with $M$, is carried out, and the outcomes are coarse-grained. Given a degenerate observable, there exist, however, infinitely many refinements. Consequently, the post-measurement state will not only depend solely on the pre-measurement state $\rho_A$ (which we can assume to be pure, without loss of generality) and the observed degenerate outcome $P^A_x$, but also on the specific refinement chosen. More precisely, letting $\delta = \tr{( P^A_x )}$ denote the ``cardinality'' of the degeneracy of outcome $P^A_x$, the updated state depends on the subset $\mathcal{B}_x = \{ P^A_{x_i} \}_{i=1}^{\delta } \subseteq M^{\prime}$ spanning the degenerate subspace $\Hs_x = P^A_x \Hs$. The state-update for non-composite systems thus obeys
\begin{equation} \label{eq: von Neumann rule for degenerate outcomes}
    \uruleA^{\mathsf{vN}}_A \left( P^A_x, \rho_A, \mathcal{B}_x \right) =  \sum_{i=1}^{\delta} \uruleL_A \left( P^A_{x_i}, \rho_A \right) \, ,
\end{equation}
where the right-hand side describes a proper mixture.
As a result, von Neumann's postulate---despite being, like L{\"u}ders' original postulate ($\mathsf{L}$), concerned only with the state-update of the measured system---violates context-independence A2 and, therefore, does not give rise to a valid update rule.\footnote{One could, however, modify \eqref{eq: von Neumann rule for degenerate outcomes} into a \textit{valid} update rule by, for example, fixing a preferred set $\mathcal{B}_x$ for each experiment yielding the outcome $P^A_x$, and subsequently extend the map to composite systems via complete positivity.}

Finally, we present a simple example of an update rule that violates local covariance A3. Assume that any measurement on a system causes the pre-measurement state to undergo a fixed, observer-independent and non-trivial\footnote{If $U=\id$, we recover the passive measurements of Sec.\ \ref{sec: pQT}, which are consistent with A3.} unitary transformation $U \in \mathcal{U}(\Hs)$, 
\begin{equation} \label{eq: unitary update rule}
    \uruleA^{\mathsf{U}}_{AB} \left( P^A_x, \,  \rho_{AB} \right) = \text{Tr}\left( P_x^A \text{Tr}_B \left( \rho_{AB} \right) \right) \, \left(U \otimes \id_B \right) \, \rho_{AB} \, \left( U^{\dagger} \otimes \id_B \right) \, .
\end{equation}
The rule $\uruleA^{\mathsf{U}}_{AB}$ is not locally covariant when using local unitaries $U_A$ that do not commute with $U$: different observers would provide inconsistent descriptions of the same measurement.

\section{How to recover the projection postulate} \label{sec: singling out Lueders}
We have set up a framework of hypothetical update rules that do give rise to consistent foils of quantum theory, and we explored their properties. It is natural to reverse our approach by turning to the question of singling out the \Luders rule of quantum theory among the set of valid updates. The correlation-free version of quantum theory (Sec.\ \ref{Correlation-freeQM})---as well as any suitable modification of von Neumann's projection \eqref{eq: von Neumann rule for degenerate outcomes}---shows that the \Luders rule cannot be recovered uniquely from assuming deterministic repeatability alone. Similarly, passive quantum theory (Sec.\ \ref{sec: pQT}) provides another example of an update rule satisfying ideality.

\Luders presents two main objections to von Neumann's projection postulate \cite{luders_uber_1950}. Firstly, he argues that the post-measurement state should depend only on the outcome and the initial state. This requirement is equivalent to context-independence A2. Secondly, he suggests that 
\begin{quote}
``the measurement of a highly degenerate quantity permits only relatively weak assertions regarding the considered ensemble. For that reason, the resulting change in state should likewise be small'' (\cite{luders_concerning_2006}, p.\ 665).
\end{quote}
In the limit of the \textit{trivial observable} represented by the identity operator $\id$, a measurement reveals no information about the original state and it should have no effect on it. The quoted property describes a form of trade-off between the applied disturbance caused by a measurement and the information acquired from the outcome. The following operational property for non-composite systems, which we term \textit{coherence}, implies the desired behaviour for measurements of the trivial observable.

\begin{defn} [\emph{Coherence}] \label{def: coherence}
    An update rule for non-composite systems $\uruleA_A$ satisfies \textit{coherence} if the outcome probability distribution of any observable $M$ is not affected by a prior measurement of a coarse-graining of $M$ (i.e.\ an observable for which $M$ is a refinement). Mathematically, for every  $P^A_x, P^A_y \in \mathcal{P}(\Hs_A)$ such that $P^A_x \, P^A_y = P^A_y \, P^A_x = P^A_x$, and every state $\rho_A \in \bar{\mathcal{S}}(\Hs_A)$, we must have
    \begin{equation} \label{eq: coherence}
        \tr{\left[ P^A_x \, \uruleA_A \left( P^A_y, \rho_A \right) \right]} = \tr{\left( P^A_x \rho_A \right)} \, .
    \end{equation}
\end{defn}

Consider an infinite ensemble prepared in an arbitrary quantum state, and examine the following two scenarios: $(i)$ a measurement of $M$ is performed on each element of the ensemble; $(ii)$ a measurement of a coarse-graining of $M$ followed by a measurement of $M$ are performed on each element of the ensemble. The property of coherence stipulates that the outcome probabilities observed when measuring $M$ in scenarios $(i)$ and $(ii)$ must be identical.
In fact, Eq.\ \eqref{eq: coherence} ensures that, for all quantum states and outcomes $P^A_x$ and $P^A_y$ such that $P^A_x \, P^A_y = P^A_y \, P^A_x = P^A_x$, the relation 
\begin{equation} \label{eq: probability-based relation for coherence}
    \text{prob}(o_1=P^A_y) \, \text{prob}(o_2=P^A_x|o_1=P^A_y) = \text{prob}(o_1=P^A_x)
\end{equation}
is satisfied, where $o_1$ and $o_2$ denote the first and second outcomes of two consecutive measurements, respectively. Here, $P^A_x$ denotes the outcome of the $M$ measurement, while $P^A_y$ denotes the outcome of the coarse-graining. If $P^A_x \neq P^A_y$, then $\Hs_x \subset \Hs_y$ (where $\Hs_i = P_i^A \Hs$), meaning that $P^A_y$ provides less information about the original state than $P^A_x$. As a result, the effect on the state cannot be arbitrarily strong but must ``preserve coherence'' in the subspace $\Hs_y$, in the sense that $\text{prob}(o_2=P^A_x|o_1=P^A_y) = \text{const.} \times \text{prob}(P^A_x)$ for all $P^A_x \in \mathcal{P}(\Hs_y)$. In the limiting case of $P^A_y = \id_A$, one obtains $\uruleA_A \left( \id_A, \rho_A \right) \propto \rho_A$, in agreement with L{\"u}ders' remark.

In some form or another, the assumption of coherence has appeared already in the literature, both in the context of quantum theory \cite{bell_moral_1966, herbut_compatibility_2004, khrennikov_neumann_2009} and of generalised probabilistic theories \cite{kleinmann_sequences_2014}. For measurements on non-composite quantum systems, coherence entails both \emph{\textit{deterministic repeatability}} \cite{kleinmann_sequences_2014} and \emph{\textit{ideality}}.

\begin{lem} \label{lem: coherence implies DR and ID}
    A coherent update rule for non-composite systems $\uruleA_A$ satisfies both deterministic repeatability and ideality.
\end{lem}
\begin{proof}
    Consider an arbitrary projector $P_x^A$ and a system residing in some state $\rho_A \in \mathcal{S}(\Hs_A)$ such that $\text{Tr} (P_x^A \rho) \neq 0$. Letting $P_y^A=P_x^A$ in Eq.\ \eqref{eq: coherence} leads to $\text{Tr} \left( P_x^A \, \uruleA_A \left( P_x^A, \rho_A \right) \right) = \text{Tr} \left( P_x^A \rho_A \right) $,
    which implies that, if repeated, the measurement will yield the same outcome with probability $1$. In other words, $\uruleA_A \left( P_x^A, \rho_A \right)$ has support in the subspace $\Hs_x = P_x^A \Hs_A$ only, i.e. $\uruleA_A \left( P_x^A, \rho_A \right) \in \bar{\mathcal{S}}(\Hs_x)$. Therefore, $\uruleA_A$ satisfies deterministic repeatability.

    Following the same argument, we know that $\uruleA_A \left( P_y^A, \rho_A \right)$ has support in $\Hs_y$ only. Let the initial state $\rho_A$ also have support in $\Hs_y$ only, i.e.\ $\text{Tr} (P_y^A \rho_A) = 1$, and consider an informationally complete set of observables on $\Hs_y$. If $\uruleA_A \left( P_y^A, \rho_A \right) \neq \rho_A$, then for at least one of the outcomes of an observable in the set, say $P_x^A \in \mathcal{P}(\Hs_y)$ (which must satisfy $P^A_x \, P^A_y = P^A_x$), we have $\tr{\left( P_x^A  \uruleA_A \left( P_y^A, \rho_A \right) \right)} \neq \tr{\left( P_x^A \rho_A \right)}$, in contradiction with \eqref{eq: coherence}. Therefore, the update rule does \emph{not} disturb the states of measured systems when the outcome is certain, i.e.\ $\uruleA_A$ satisfies ideality.
\end{proof}

Indeed, coherence alone is sufficient to \textit{derive} the single-system \Luders rule $\uruleL_A$, i.e.\ the original projection postulate ($\mathsf{L}$). This argument was first presented in 1966 \cite{bell_moral_1966} and similar versions of it have since been rediscovered several times; see Appendix \ref{sec:historical review Luders} for a summary. Below, we present a simple proof of this result.
\begin{thm} \label{thm: coherence and Luders}
    For non-composite systems, the \Luders rule $\uruleL_A$ is the only coherent update rule.
\end{thm}
\begin{proof}
    Substituting $ \uruleL_A$ in \eqref{eq: coherence} leads to
    \begin{equation}
        \tr{\left( P^A_x  P^A_y  \rho_A \, P^A_y \right)} = \tr{\left( P^A_x  \, \rho_A \right)} \, ,
    \end{equation}
    which holds for all $\rho_A$ and $P^A_x \, P^A_y = P^A_y \, P^A_x = P^A_x $. Hence the \Luders projection for single systems is coherent.
    
    To show the converse, suppose $\uruleA_A$ satisfies coherence but $\uruleA_A \neq \uruleL_A$. Then, there must exist some outcome $P^A_y \in \mathcal{P}(\Hs_A)$ and state $\rho_A \in \mathcal{S}(\Hs_A)$ such that $\uruleA_A \left( P^A_y, \rho_A \right) \neq \uruleL_A \left( P^A_y, \rho_A \right)$. Since both rules are coherent, hence deterministically repeatable by Lemma \ref{lem: coherence implies DR and ID}, both updated states must lie within the same subspace defined by $P^A_y$, i.e.\ $\uruleA_A \left( P^A_y, \rho_A \right)$, $\uruleL_A \left( P^A_y, \rho_A \right) \in \bar{\mathcal{S}}(\Hs_y)$. Consider a set of observables of $\Hs_y$ which allow one to reconstruct any quantum state in $\bar{\mathcal{S}}(\Hs_y)$ tomographically. The states $\uruleA_A \left( P^A_y, \rho_A \right)$ and $\uruleL_A \left( P^A_y, \rho_A \right)$ must yield different probabilities for at least one outcome $P^A_x \in \mathcal{P}(\Hs_y)$ in the set,
    \begin{equation}
        \tr{\left[ P^A_x \, \uruleA_A \left( P^A_y, \rho_A \right) \right]} \neq \tr{\left[ P^A_x \, \uruleL_A \left( P^A_y, \rho_A \right) \right]} \, .
    \end{equation}
    Substituting $\uruleL_A \left( P^A_y, \rho_A \right) = P^A_y \, \rho_A \, P^A_y $ and using the fact that $P^A_x P^A_y = P^A_y \, P^A_x = P^A_x$, we obtain
    \begin{equation}
         \tr{\left[ P^A_x \, \uruleA_A \left( P^A_y, \rho_A \right) \right]} \neq \tr{\left( P^A_x \rho_A \right)} \, ,
    \end{equation}
    which contradicts the assumption of coherence. We thus conclude that $\uruleA_A= \uruleL_A$.
\end{proof}
The information-disturbance trade-off principle, as encapsulated by the property of coherence, completely characterises the update rule for quantum measurements on non-composite systems. However, there do exist GUTs with different update rules for multi-partite systems that reduce to the single-system \Luders rule. The correlation-free theory of Sec.\ \ref{Correlation-freeQM} defined by the update rule $\urule^{\mathsf{locL}}$ is an explicit example. Straightforward generalisations of coherence to composite systems do not seem to be sufficiently strong to generalise the proof of Theorem \ref{thm: coherence and Luders}. For example, the assumption that a sequence of measurements of $P_y^A$ and $P_x^A$ (see Def.\ \ref{def: coherence}) on subsystem $\Hs_A$ leads to the same post-measurement state for $\Hs_{AB}$ as a single measurement of $P_x^A$ fails to rule out correlation-free quantum theory. 

Additional assumptions are necessary to fully recover the \Luders rule $\uruleL_{AB}$ for composite systems. One such assumption is composition compatibility (Def.\ \ref{def: composition compatibility})---see Theorem \ref{thm: Luders from coherence and CC} below. This final step in deriving the projection postulate is often overlooked \cite{bell_moral_1966, herbut_compatibility_2004, martinez_search_1990}, as the extension of $\uruleL_A$ to local measurements on composite systems is typically assumed implicitly, despite not being stated in or implied by the standard postulates.

\begin{thm} \label{thm: Luders from coherence and CC}
    The \Luders rule $\uruleL$ is the unique generalised state-update rule that satisfies both coherence and composition compatibility.
\end{thm}
\begin{proof}
    Theorem \ref{thm: coherence and Luders} derives the \Luders projection $\uruleL_A$ for non-composite systems from coherence alone. Then, the quantum-mechanical extension to composite systems $\uruleL_{AB}$ is fixed by Eq.\ \eqref{eq: def CC} defining composition compatibility, as already discussed in Sec.\ \ref{sec: quantum theory (extensions)}.
\end{proof}

This derivation of the \Luders rule does \textit{not} rely on assuming linearity of the update rule. There may, of course, be other physically appealing operational principles ruling out all update rules but the \Luders rule. For instance, it seems promising to explore the implications of \textit{nonlocality}. Could quantum theory be singled out by imposing specific constraints on the degree of nonlocality exhibited by generalised state-update theories?

\section{Conclusions} \label{Sec: Discussion}

\subsection{Summary}

We have introduced a framework that allows us to systematically investigate modifications of the projection postulate. Theories with generalised state-update rules (GURs) differ from quantum theory only in assigning other post-measurement states.

Upon formalising the concept of state-update rules, we notice that the quantum mechanical L{\"u}ders rule is typically presented in an operationally incomplete way. Spelling out the projection postulate \eqref{eq: normalised collapsed state} for \textit{single} systems does not entail a unique update rule when local measurements are performed on \textit{composite} systems. In standard quantum theory, the necessary extension from single to composite systems is normally assumed implicitly to have a specific form. This observation is a simple example of the subtleties in the definition of update rules. The specific update rule used in composite quantum systems is, of course, well motivated by experimentally verified correlations between measurement outcomes obtained for entangled states. 

Any valid state-update rule must satisfy six core requirements which include the uniqueness of post-measurement states, no-signalling, and consistency for measurements carried out by observers with access to subsystems only. We have shown that quantum theory and passive quantum theory \cite{fiorentino_quantum_2023} satisfy the requirements. Three new examples of consistent generalised state-update rules have been defined, each one leading to a toy theory structurally different from quantum theory.

\textit{Correlation-free quantum theory} turns out to be \textit{locally} indistinguishable from standard quantum theory since the update rule for \textit{single} systems coincides with the projection postulate. However, the proposed update rule does not extend in the standard way to \textit{composite systems}. The update rule is ``locally L{\"u}ders'' but local measurements performed on entangled states do not create the correlations known from quantum theory. 

\textit{Depolarising measurements} are described by an update rule that sends any pre-measurement state to the maximally mixed state, irrespective of the measurement outcome. The update effectively acts as a completely depolarising channel, flattening the probability distribution for subsequent outcomes of measurements. No information about the pre-measurement state is retained.  

\textit{Probability-amplifying projective rules} lead to post-measurement states that assign a larger probability to the measurement outcome than the pre-measurement state did. The standard projective update rule is the extreme case of this rule: the outcome of the second measurement is necessarily the same as that of the first one, implying deterministic repeatability. The \Luders rule for single systems can be approximated arbitrarily well by probability-amplifying projective state updates.

To avoid signalling in theories with depolarising or probability-amplifying update rules, careful adjustments are needed when extending them to composite systems. The ways in which state-update rules may violate core requirements have been illustrated by a number of (ultimately invalid) GURs.

Having established a variety of acceptable update rules, it is natural to look for properties that single out the \Luders rule. We have shown that the property of  \textit{coherence} uniquely implies the \Luders rule for non-composite systems. This result agrees with earlier ones based on similar concepts describing a form of information-disturbance trade-off. However, our systematic approach to GURs highlights the fact that many derivations of the \Luders rule may be considered incomplete by not explicitly addressing \textit{composite} systems. We find that coherence on its own is \textit{insufficient} to single out the projection postulate when considering composite systems; equally, composition compatibility alone is not strong enough to do so. However, combining it with coherence implies the projective \Luders rule for composite systems, too (Theorem \ref{thm: Luders from coherence and CC}).

\subsection{Discussion and outlook}

The framework of generalised state-update rules we present affords us with a bird's-eye view of the projection postulate and physically reasonable modifications thereof, their basic features and their implications. It is useful to think of update rules as generalisations of quantum instruments that are neither linear nor completely positive, all the while still providing consistent descriptions of measurements on (hypothetical) quantum systems.

Non-linear update rules may enable observers to distinguish between proper and improper mixtures. In the corresponding foil theories, density operators will \textit{not} provide a complete description of the hypothetical physical systems. Importantly, complete positivity is not strong enough to ensure that such measurement-induced transformations assign post-measurement states  to composite systems \textit{consistently}. In the words of \cite{czachor_complete_1998}, complete positivity is ``physically unfitting'' for this purpose.

Reconstructions of quantum theory that engage---either directly or indirectly---with post-measure\-ment states usually \textit{assume} measurement-induced state transformations to be linear in the pre-measure\-ment state \cite{barrett_information_2007, kleinmann_sequences_2014, masanes_measurement_2019}. This assumption, often left implicit \cite{stacey_masanes-galley-muller_2022, kent_measurements_2023},  \textit{a priori} excludes measurement behaviours such as the trivial \textit{non-collapsing} rule of  ``passive quantum theory'' that are, in fact, operationally well-defined and consistent with fundamental principles such as no-signalling (see Sec.\ \ref{subsec: assumptions we make}). Preliminary results based on the framework of GURs and earlier work (e.g.\  \cite{ozawa_quantum_1984, flatt_gleason-busch_2017}) suggest that the linearity of update rules could be \textit{derived} from natural operational assumptions, 
\textcolor{modification}{mirroring efforts to establish the linearity of the quantum time evolution \cite{gisin_weinbergs_1990, wilson_origin_2023}}.

It seems promising to investigate the role of \textit{nonlocality} from the perspective of theories with generalised update rules. Do other update rules exist that are capable of replicating the correlations found in quantum theory? What is more, can they lead to super-quantum correlations?

\subsection*{Acknowledgements} The authors are grateful for support through grant RPG-2024-201 by the Leverhulme Trust.

\appendix

\section{Derivations of the \Luders rule -- A brief survey} \label{sec:historical review Luders}

Derivations of the \Luders rule have a long history, with similar arguments being made over the years. We will briefly review derivations known to us, highlighting the assumptions made. A key observation is that most of the justifications focus on non-composite systems, either ignoring measurements on composite systems or making implicit assumptions about them.

In 1966, Bell and Nauenberg showed that the quantum mechanical state-update can be obtained by assuming that the outcome probability distribution of a common refinement for commuting observables $M_1$ and $M_2$ equals the joint probability distribution obtained by measuring $M_1$ and $M_2$ sequentially, in any order \cite{bell_moral_1966}. This assumption is mathematically equivalent to that of \textit{coherence} (Def.\ \ref{def: coherence}) employed in the proof of Theorem \ref{thm: coherence and Luders} in Sec.\ \ref{sec: singling out Lueders}. To the best of our knowledge, the only reference to this early operational derivation of the \Luders rule in non-composite systems is by Herbut \cite{herbut_compatibility_2004}. Over the last 40 years, the argument by Bell and Nauenberg has been rediscovered at least three times, using the concept of coherence.

In 1983, Cassinelli and Zanghì formulated quantum theory in terms of a generalised probability space and showed how to recover the single-system \Luders rule in this context. To do so, they describe the state-update following measurements in terms of a generalised conditional probability defined on the generalised probability space---i.e.\ a probability measure on $\mathcal{P}(\Hs)$ compatible with Born's rule and satisfying ``coherence'' \cite{cassinelli_conditional_1983}.

In 2009, Khrennikov showed that the \Luders rule for single systems follows from von Neumann's projection postulate for non-degenerate observables if ``coherence'' is assumed \cite{khrennikov_neumann_2009}. Claiming that ``this important observation remained unnoticed for the past 70 years'', the author appears to be unaware of Bell and Nauenberg's work.

In 2014, Kleinmann developed a formalism for sequential measurements in a broad class of generalised probabilistic theories (GPTs), introducing a generalisation of the \Luders rule based on coherence \cite{kleinmann_sequences_2014}. A simple derivation of the non-composite \Luders rule based on coherence is also given, without reference to earlier work. When applied to quantum theory, i.e. in a  Hilbert-space setting for states, Kleinmann's approach reproduces  \textit{convex-linear update rules} for single systems. This restriction is due to assuming \textit{preparation indistinguishability} throughout (see Sec.\ \ref{sec: quantum theory (extensions)}), thereby excluding operationally valid alternatives such as the passive update rule from the outset.

Herbut actually presents \textit{two} distinct derivations applying to measurements on non-composite systems in a paper from 1969 \cite{herbut_derivation_1969}. In an information-theoretic approach, he postulates (i) deterministic repeatability (see Sec.\ \ref{Qmstateupdaterule}) and that (ii) the post-measurement state minimises its distance from the pre-measurement state. The \Luders rule is then shown to  map the pre-measurement state to the least distinguishable state compatible with outcome repeatability. This result remained largely unnoticed, while anticipating later arguments deriving the state-update from minimisation principles, including those by Marchand \cite{marchand_statistical_1977}, Dieks \cite{dieks_distance_1983}, and Hadjisavvas \cite{hadjisavvas_distance_1981}. More recently, the single-system \Luders rule was obtained from minimising non-symmetric quantum information distances \cite{kostecki_luders_2014, hellmann_quantum_2016}, thereby establishing a mathematical connection to the Bayes-Laplace rule. Ozawa \cite{ozawa_quantum_1984, ozawa_quantum_1998} derives the quantum instrument formalism describing state-updates from Bayes' rule by assuming the standard joint probability formula for local quantum measurements. In a sense, the author adopts a \textit{stronger} interpretation of postulate $(\mathsf{B})$ (cf.\ Sec.\ \ref{sec: axioms QT}) than the one used explicitly by us, and implicitly by other authors \cite{kent_nonlinearity_2005, aaronson_space_2016, kent_quantum_2021, kent_measurements_2023} where $(\mathsf{L})$ is modified. The stronger interpretation of Born's rule is assumed to govern  individual measurement outcomes as well as \emph{correlations} between \textit{local} measurements carried out by different experimenters with access to parts of a composite system only. Essentially, it subsumes the property of local tomography (cf.\ Sec.\ \ref{sec: quantum theory (extensions)}). 

Herbut's second derivation in \cite{herbut_derivation_1969} takes an operational approach. In addition to deterministic repeatability, it is assumed that performing an unconditional measurement of observable $M$ does not affect the expectation values of any observables commuting with $M$. The argument focusses on measurements on single systems. However, Herbut also briefly considers composite systems, discussing in particular how local \Luders measurements modify joint states without enabling signalling. When doing so, he implicitly adopts the standard extension to composite systems (see Sec.\ \ref{sec: quantum theory (extensions)}). What is more, Herbut takes convex-linearity over states to be a fundamental requirement for state updates, without justifying this assumption. 

Herbut revisits the question in a 1974 paper \cite{herbut_minimal_1974}. Here, the \Luders rule for single systems is obtained by combining the assumption of deterministic repeatability with the ``preservation of sharp values''. An observable $M$ has a \textit{sharp value} $x$ if  $\text{Tr}(P_x^A \rho )=1$ holds for some $P_x \in M$. This value is \textit{preserved} if, given an initial state $\rho$ with a sharp $M$-value $x$, a measurement of any observable commuting with $M$ leaves unchanged the value $x$, regardless of the observed outcome. When proving this result for improper mixtures, Herbut implicitly adopts ``composition compatibility'' to describe local measurements. 

Martinez \cite{martinez_search_1990}  offers an early account of known derivations of the \Luders rule including Herbut's work but not Bell and Nauenberg's contribution. The paper also demonstrates that the \textit{converse of \textit{Lüders theorem}} \cite{luders_uber_1950, busch_quantum_1996} implies the quantum mechanical collapse for non-composite systems: if, for any pair of commuting observables $M_1$ and $M_2$ and any initial state, we assume that the sequence of measurements $M_1 \prec M_2 \prec M_1$ always produces identical outcomes for $M_1$ regardless of the outcome of the $M_2$ measurement, then the post-measurement state must necessarily be given by the \Luders rule. Martinez further examines a quantum-logical derivation employing the ``Sasaki hook'' presented in \cite{friedman_quantum_1978}. The extension of the update rule to local measurements on subsystems is not considered.

The proof that establishes the \Luders instrument as the unique ideal quantum instrument for sharp discrete measurements dates back to Davies \cite{davies_quantum_1976}, as recognised in \cite{busch_quantum_1996, busch_repeatable_1995}. The notion of  a quantum instrument incorporates two key assumptions: convex-linearity of the single-system state-update rule, and the standard linear extension to composite systems enabled by the complete positivity of $\uruleA_A$ (see Eq.\ \eqref{eq: extension Luders (CP)} in Sec.\ \ref{sec: quantum theory (extensions)}). This argument for isolating the \Luders collapse for single systems has been rediscovered occasionally (see \ \cite{marinkovic_note_1983}, for example).

In 2019, Masanes et al.\ advance the argument that both the Born rule and the post-measurement state rule can be derived from the other postulates of quantum theory when supplemented by suitable operational constraints \cite{masanes_measurement_2019}. In doing so, they consider alternative rules for computing outcome probabilities, explicitly discussing both composite systems and local measurements. Having established the uniqueness of the Born rule within their framework, they assert that quantum instruments constitute the only consistent description of state updates. This conclusion, along with their broader framework, has attracted some criticism \cite{stacey_masanes-galley-muller_2022, galley_reply_2022, kent_measurements_2023, masanes_response_2025, stacey_contradictions_2024}. The argument for quantum instruments depends on an assumption regarding sequential measurements that, as already shown in \cite{flatt_gleason-busch_2017}, is mathematically equivalent to requiring convex-linearity of state updates. 

Conceptually, Masanes et al.'s approach mirrors Kleinmann's framework \cite{kleinmann_sequences_2014} in that it is limited \textit{a priori} to update rules satisfying preparation indistinguishability. This restriction explains why they adopt the completely positive extension \eqref{eq: extension Luders (CP)} for local measurements on composite systems: any other extension would necessarily violate their assumption of preparation indistinguishability.

In retrospect, the \Luders rule is seen to have been derived on the basis of operational, information-theoretic, quantum-logical, and epistemically motivated assumptions. Most of the derivations are limited to single-system update rules, and their extension to local measurements on subsystems is either ignored or tacitly assumed. The extension of state-update rules is, however,  not trivial, since the single-system rule alone does \textit{not} determine it. Additional structural assumptions such as preparation indistinguishability or composition compatibility  (cf.\ Theorem \ref{thm: Luders from coherence and CC}) are required to exclude alternative update rules in the multi-partite case.

\printbibliography[heading=bibintoc]



\end{document}